\documentclass[11pt,draftcls,onecolumn]{IEEEtran}
\usepackage{graphicx}
\usepackage{epstopdf}
\usepackage{placeins}
\usepackage{array}
\DeclareGraphicsExtensions{.eps,.pdf} 
\graphicspath{ {figures/} }
\usepackage{cite}
\usepackage{balance}
\usepackage{amssymb,amsmath}
\usepackage{subfigure}
\usepackage[ruled,vlined]{algorithm2e}
\usepackage{amssymb,amsmath,mathtools}
\usepackage{amssymb,amsmath}
\usepackage{algorithmic}
\usepackage{color}
\usepackage{multirow}
\usepackage{stackrel,amssymb,amsmath}
\usepackage{empheq}
\usepackage{cases}
\makeatletter
\newcommand\restartchapters{\par
  \setcounter{chapter}{0}%
  \setcounter{section}{0}%
  \gdef\@chapapp{\chaptername}%
  \gdef\thechapter{\@arabic\c@chapter}}
\makeatother

\def\bq{{\bf q}}

\def\bw{{\bf w}}
\def\bq{{\bf q}}

\newtheorem{theorem}{\bf Theorem}

\newtheorem{lemma}{\it \underline{Lemma}}

\newtheorem{proposition}{\it \underline{Proposition}}

\renewcommand{\algorithmicrequire}{\textbf{Input:}}

\newcommand \ile {\stackrel{\mathclap{\normalfont\mbox{i}}}{\le}}
\newcommand \iieq {\stackrel{\mathclap{\normalfont\mbox{i2}}}{=}}
\newcommand \iiile {\stackrel{\mathclap{\normalfont\mbox{i3}}}{\le}}
\newcommand \iiiile {\stackrel{\mathclap{\normalfont\mbox{i4}}}{\le}}

\newcommand \ifiveeq {\stackrel{\mathclap{\normalfont\mbox{i5}}}{=}}
\newcommand \isixle {\stackrel{\mathclap{\normalfont\mbox{i6}}}{\le}}
\newcommand \isevenle {\stackrel{\mathclap{\normalfont\mbox{i7}}}{\le}}

\makeatletter
\g@addto@macro\normalsize{%
 \setlength\abovedisplayskip{4pt}
 \setlength\belowdisplayskip{4pt}
 \setlength\abovedisplayshortskip{4pt}
 \setlength\belowdisplayshortskip{4pt}
}
\makeatother

\makeatletter
\def\endthebibliography{%
	\def\@noitemerr{\@latex@warning{Empty `thebibliography' environment}}%
	\endlist
}
\makeatother

\begin{document}
\bstctlcite{IEEEexample:BSTcontrol}
\title{\LARGE Throughput Maximization for Wireless Communication systems with Backscatter- and Cache-assisted UAV Technology}
\author{ \normalsize
\IEEEauthorblockN{$\text{Dinh-Hieu Tran}, \textit{Student Member, IEEE}$, $\text{Sumit Gautam}, \textit{Member, IEEE}$, $\text{Symeon Chatzinotas}, \textit{Senior Member, IEEE},$ $\text{and Bj{\"o}rn Ottersten}, \textit{Fellow, IEEE}$ }
\thanks{}
\thanks{}
\vspace*{-0.5cm}}
\maketitle
\thispagestyle{empty}
\pagestyle{empty}
\vspace*{-1cm}
\begin{abstract}
Unmanned aerial vehicle (UAV) has been widely adopted in wireless systems due to its flexibility, mobility, and agility. Nevertheless, a limited onboard battery greatly hinders UAV to prolong the serving time from communication tasks that need a high power consumption in active RF communications. Fortunately, caching and backscatter communication (BackCom) are appealing technology for energy-efficient communication systems. This motivates us to investigate a wireless communication network with backscatter- and cache-assisted UAV technology. We assume a UAV with a cache memory is deployed as a flying backscatter device (BD), term the UAV-enabled BD (UB), to relay the source's signals to the destination. Besides, the UAV can harvest energy from the source's RF signals and then utilizes it for backscattering information to the destination. In this context, we aim to maximize the total throughput by jointly optimizing the dynamic time splitting (DTS) ratio, backscatter coefficient, and the UB's trajectory with caching capability at the UB corresponding to linear energy harvesting (LEH) and non-linear energy harvesting (NLEH) models. These formulations are troublesome to directly solve since they are mixed-integer non-convex problems. To find solutions, we decompose the original problem into three sub-problems, whereas we first optimize the DTS ratio for a given backscatter coefficient and UB's trajectory, followed by the backscatter coefficient optimization for a given DTS ratio and UB's trajectory, and the UB's trajectory is finally optimized for a given DTS ratio and backscatter coefficient. By using the KKT conditions, closed-form expressions for the optimal values of the DTS ratio and backscatter coefficient are obtained which greatly reduce the computation time. Moreover, the solution of the third sub-problem can be acquired by adopting the successive convex approximation (SCA) technique. Consequently, efficient alternating algorithms are proposed for both EH models by leveraging the block coordinate descent (BCD) method. Finally, the intensive numerical results demonstrate that our proposed schemes achieve significant throughput gain in comparison to the benchmark schemes.
\end{abstract}

\begin{IEEEkeywords}
\vspace*{-0.2cm} Backscatter communication (BackCom), Caching, energy harvesting, reflection coefficient control, time allocation, trajectory design, unmanned aerial vehicle (UAV).
\end{IEEEkeywords}

\vspace*{-0.5cm}
\section{Introduction} \label{Introduction}

Unmanned aerial vehicles (UAVs) have attracted significant attention from both academia and industry due to their flexible deployment, low cost, and high maneuverability \cite{Mozaffari2019Survey,Yeng2019Sur}. Indeed, UAVs have enabled various applications such as military, agriculture, transportation, search and rescue missions, surveillance and monitoring, telecommunications \cite{Mozaffari2019Survey,Yeng2019Sur,Hieu,Tran2020FDUAV,Y_Zeng_1,Mozaffari,hua2019}. Particularly, if properly designed and deployed, UAVs can provide efficient solutions for wireless communication networks. Specifically, UAVs can be utilized as aerial/flying base stations (BSs) to support terrestrial BSs that are located in fixed locations and cannot be shifted elsewhere. Especially, in a natural disaster where terrestrial BSs are damaged or isolated, portable BSs do exist but they have to be moved using ground vehicles which is problematic when infrastructures for publication transportation systems may be destroyed. Consequently, UAVs can be swiftly deployed to disseminate vital information to people or help them to communicate with authorities as soon as possible \cite{Hieu,Tran2020FDUAV}. For industrial applications, Google Project Wing and Amazon Prime Air have built and tested drones deliveries that could be used after a disaster (i.e., flood, earthquake) or in extreme weather conditions \cite{stewart2014google}. They expect to develop an advanced delivery system where drones help to bring medications or foods to people in the areas that conventional vehicles can not reach. Besides, Facebook Halts Aquila and Google Loon projects aim at beaming internet access to people around the world who can not connect to the Internet by using drones/balloons \cite{kelion2015facebook}. Furthermore, AT$\&$T and Qualcomm are planning to adopt UAVs for facilitating large-scale wireless communications in 5G networks \cite{LTE_Drone}. 

Recently, UAVs have been proposed as relays to improve the connectivity of networks \cite{ZhongUAVRelay,Li2020UAV,Sharma2020UAVSat,Sun2019UAVsecure,Tran2020FDUAV,Ye2020UAV}. Especially, in case direct communications links are missing due to shadowing or un-communication devices by the BSs during peak hours. In these cases, UAVs are deployed as relays to help convey information from the source to the destination. In \cite{ZhongUAVRelay}, the authors studied UAVs-assisted self-organized device-to-device (D2D) networks. Specifically, they aimed to maximize the total throughput via jointly optimizing the channel allocation, relay deployment, and relay assignment. Li et al. \cite{Li2020UAV} investigated the joint positioning and power control to maximize the sum rate of UAV relay networks, wherein the UAV utilized two-way communications between the BS and a set of users. The works in \cite{Sharma2020UAVSat} and \cite{Sun2019UAVsecure} investigated the secure transmission in UAV relay networks. Sharma et al. in \cite{Sharma2020UAVSat} proposed a novel secure 3D UAV relaying for hybrid satellite-terrestrial networks (HSTNs) in the presence of a flying eavesdropper and then they investigate secrecy outage probability and the probability of non-zero secrecy capacity. Sun et al. \cite{Sun2019UAVsecure} studied secure transmissions of
millimeter-wave simultaneous wireless information and power
transfer (SWIPT) UAV relay networks with multiple eavesdroppers. In contract to \cite{ZhongUAVRelay,Li2020UAV,Sharma2020UAVSat,Sun2019UAVsecure} that only considered half-duplex (HD), \cite{Tran2020FDUAV} and \cite{Ye2020UAV} investigated the rotary-wing UAV-enabled FD Internet-of-Things (IoT) networks.

Notably, the UAVs in the above works emit active RF signals to the destination requiring high energy consumption which reduces the lifetime of UAVs with a limited onboard battery. In this regard, backscatter communication (BackCom) is a promising solution since a typical backscatter circuit's power consumption is usually in the order of $\mu W$ \cite{Lu2018Back,Huynh2018Back}, which is significantly lower than that of active RF transmission, i.e., in Watts. Consequently, BackCom has recently emerged as a key concern for UAV communication networks \cite{yang2019energy,farajzadeh2019uav,Farajzadeh2020Back,Hu2020Back,hua2019}. In \cite{hua2019}, the authors proposed two novel schemes termed the transmit-backscatter protocol and transmit-backscatter relay protocol corresponding to the presence or absence of a direct link between backscatter user and receiver in UAV-aided BackCom networks. Yang et al. \cite{yang2019energy} considered a UAV-aided BackCom network comprising of backscatter devices (BDs) and carrier emitters (CEs) that are randomly distributed on the ground. They aimed at maximum energy efficiency (EE) by jointly optimizing the BDs' scheduling, the UAV's trajectory, and the CEs' transmit power. Farajzadeh et al. \cite{farajzadeh2019uav} proposed a novel UAV data collection in NOMA BackCom networks, where the UAV acted both as a power source and a data collector. The objective was to jointly design several backscatter devices, UAV’s altitude, and backscatter coefficient to maximize the total successfully decoded bits while minimizing the UAV's flight time. The same authors in \cite{Farajzadeh2020Back} studied the first work that considered UAV as an enabler to improve over-the-air computation (AirComp)'s performance. Hu et al. \cite{Hu2020Back} proposed the first work that investigated secure transmissions in UAV-aided BackCom networks. Despite prominent achievements in UAV-assisted BackCom networks in \cite{yang2019energy,farajzadeh2019uav,Farajzadeh2020Back,Hu2020Back,hua2019}, aforementioned works do not take caching into consideration.

Recent works have shown that some popular files are repeatedly demanded by users, which accounts for a massive portion of data traffic \cite{Erman2011Caching,Saja2020Caching}. By storing a part of popular content in the cache of edge nodes, wireless caching is a promising method to reduce traffic load, especially during peak hours \cite{Srikan2020Caching}. Some recent works such as in \cite{Xu2018Caching,Cheng2020UAVCache,Chen2017UAVCache,Chai2020UAVCache,Wu2020UAVcache} have been recently devoted to cache-assisted UAV communications. Xu et al. \cite{Xu2018Caching} proposed a novel scheme to overcome the endurance issue at the UAV by utilizing proactive caching. Specifically, they aimed at minimizing the weighted sum of the file caching cost and the retrieval cost by jointly optimizing the UAV communication scheduling, UAV trajectory, and file caching policy. Cheng et al. \cite{Cheng2020UAVCache} proposed a novel scheme to assure the secure transmission for UAV relay networks with caching capability. The learning-based approaches in cache-enabled UAV communications were investigated in \cite{Chen2017UAVCache,Chai2020UAVCache,Wu2020UAVcache}. Chen et al. \cite{Chen2017UAVCache} proposed the first work to analyze the utilization of caching in UAV communications based on conceptor-based echo state networks (ESNs). Different from existing works that focused on finite-time horizon offline trajectory design, Chai et al. \cite{Chai2020UAVCache} proposed an online trajectory and resource allocation optimization for cache-enabled UAV wireless communications. Wu et al. \cite{Wu2020UAVcache} adopted a convolutional neural network (CNN)-based deep supervised learning scheme for pushing up the decision-making speed in the highly dynamic vehicular networks.

From the above discussions and the fact that caching and BackCom are energy-efficient communication technologies for UAV communication networks, this paper investigates a caching UAV-enabled BackCom network, in which a UAV can store a part of popular contents in its cache. Besides, the UAV is equipped with a backscatter circuit that can harvest the RF signal from the source and then use this energy for backscattering to the destination. Since backscatter devices consume a relatively low amount of power  (i.e., in $\mu W$), which is well-fitting with RF-power energy, i.e., up to 10 $\mu W$ \cite{lu2014wireless}. In contrast to the above works in \cite{yang2019energy,farajzadeh2019uav,Farajzadeh2020Back,Hu2020Back,hua2019} that only consider UAV as a transmitter/receiver, this work considers UAV as an aerial BD which harvests energy from the source's RF signal and then utilizes this energy for backscattering signal to the destination. \textit{To our best knowledge, this is the first work that jointly considers the combination of caching and UAV in BackCom networks.} In summary, our  contributions are as follows:
\begin{itemize}
	\item We propose a novel backscatter- and cache-assisted UAV communication network for LEH and NLEH models. Caching and backscatter can reduce the transmit power of the UAV and thus overcome the sustainability issue at the UAV. This is the first work that jointly considers UAV, caching, and BackCom.
	\item Most conventional works on wireless-powered UAV communications only consider the linear energy harvesting (LEH) model for ease of analysis. However, the LEH model only works well in the case of the low input power of the harvesting circuit. This motivates us to consider both LEH and non-linear EH (NLEH) models to give a full picture of the advantages/disadvantages of each scheme. Moreover, because the UB flies from initial to final locations, it cannot hover over the source all the time. Thus, the backscatter coefficient and DTS ratio should be carefully designed to suitable with UB's position at each time slot $n$. Particularly, this exists a trade-off for the DTS ratio in each time slot which directly impacts the amount of harvested energy and backscatter rate. 
	\item Motivated by the above considerations, we formulate an optimization problem to maximize the total collected throughput at the destination, subject to constraints on the limited flying time, UB's maximum speed, UB's trajectory, DTS ratio in each time slot, and maximum value of backscatter coefficient for both LEH and NLEH models. These formulations are mixed-integer non-convex problems and challenges to be solved.
	\item We decompose the problem into three sub-problems, wherein we first optimize the DTS ratio for a given UB's trajectory and backscatter coefficient, followed by the backscatter coefficient optimization for a given DTS ratio and trajectory, and lastly, we optimize the trajectory for a given DTS ratio and backscatter coefficient. Particularly, the closed-form expressions for the DTS ratio and backscatter coefficient are derived which dramatically
	reduces the computation time. The trajectory optimization sub-problem can be solved by leveraging the successive convex approximation (SCA) technique. Based on the solutions of these three sub-problems, we propose three-layer alternating algorithms to solve formulated problems adopting the block coordinate descent (BCD) method. 
	\item The effectiveness of the proposed schemes is demonstrated via numerical results, which show significant enhancements concerning the total collected throughput at the destination in comparison to the benchmark schemes. Specifically, the benchmark schemes are designed similar to that of our proposed algorithms but without caching capability or with a fixed DTS ratio or with a fixed trajectory for both LEH and NLEH models.
 \end{itemize}

The rest of the paper is organized as follows. The system model and problem formulation are given in Section~\ref{System_Model}. The proposed iterative algorithm for solving linear EH model-based UAV-enabled BackCom is presented in Section~\ref{sec:3}. While, Section~\ref{sec:4} treats the non-linear EH model. Numerical results are illustrated in Section~\ref{Sec:Num}, and Section~\ref{Sec:Con} concludes the paper.

\emph{Notation}: Scalars and vectors are denoted by lower-case letters and boldface lower-case letters, respectively. For a set $\mathcal{{K}}$, $|\mathcal{{K}}|$ denotes its cardinality. For a vector $\bf v$, $\left\| \bf v \right\|_1$ and $\left\| \bf v \right\|$  denote its $\ell_1$ and Euclidean ($\ell_2$) norm, respectively. $\mathbb{R}$ represents for the real matrix. $\mathbb{R}^+$ denotes the non-negative real numbers, i.e., $\mathbb{R}^+=\{x \in \mathbb{R}|x \ge 0\}$. $x \sim {\cal{CN}}(0,\sigma^2)$ represents circularly symmetric complex Gaussian random variable with zero mean and variance $\sigma^2$. Finally, $\mathbb{E} [x]$ denotes the expected value of $x$.

\begin{figure}[t]
\centering
\includegraphics[width=8cm,height=6cm]{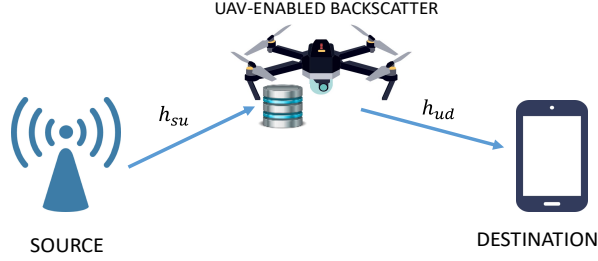}
\vspace{-1.5cm}
\caption { System model: The cache-aided UAV acts as a backscatter to convey the data from a source to a destination, , wherein the UB is quipped with a energy harvester which harvests energy from the transmit RF signal. }
\label{fig:1}   
\end{figure}
 
\vspace{-0.2cm}
\section{System Model and  Problem Formulation} \label{System_Model}
We consider a cache-assisted UAV-enabled BackCom network, where a UAV is equipped with a backscatter circuit, namely UAV-enabled backscatter device (UB), to assist the source to transmit data to the destination as shown in Fig. \ref{fig:1}. Herein, we assume that the direct transmission link from the source to the destination is impossible due to a heavy obstacle or severe fading. In this work, we focus on communication links between the source to the UB and from the UB to the destination with an assumption that all other users are successfully served by the source through terrestrial communication. Notably, non-terrestrial communication is recognized as a key component to provide cost-effective and high-capacity connectivity in future 5G and beyond/6G wireless networks \cite{giordani2019non}. The flight altitude of UAV is assumed to be fixed at $H$ meter. We assume the total flying time of UB is $T$. To make the problem tractable, the time period $T$ is equally divided into $N$ time slots of $\delta_t=T/N$. Consequently, the location of the UAV at time slot $n$ is ${\bq}_n$, with $n \in {\cal N}=\{0,\dots,N\}$. Moreover, the locations of the source and the destination are assumed to be fixed at ${\bw}_s$ and ${\bw}_d$, respectively.

\subsection{Ground-to-Air Channel Model}
By denoting $V_{max}$ as a maximum speed of the UB, the UB's constraints can be represented as
\begin{align}
\label{eq:1}
\left\| \bq_{n+1} - \bq_{n} \right\| \le \delta_d = V_{\rm max}{\delta _t},n=0,\dots,N-1.
\end{align}
\begin{align}
	\label{eq:2}
	\bq_0 = \bq_{\rm I}, \bq_N=\bq_{\rm F},
\end{align}
where $\bq_{\rm I}$ and $\bq_{\rm F}$ is the initial and final location of the UB.

For analytical convenience, let us denote the source, destination, and UB by $s$, $d$, and $u$, respectively. Consequently, the distance from $s \to u$ or $u \to d$ at time slot $n$ is given as
\begin{align}
\label{eq:3}
{d^n_{iu}} = \sqrt {{H^2} + {{\left\| {\bq_n - {\bf w}_i} \right\|}^2}} , i \in \{s, d\}, \; \forall n, 
\end{align}
where ${\bf w}_s$ and ${\bf w}_d$ are fixed locations of the source and destination.

This work considers a realistic channel model consisting of both line-of-sight (LOS) and non-line-of-sight
(NLOS) channel. This is because the UB can operate in different environments, e.g., urban, sub-urban, or rural area. Particularly, we take large-scale fading and small-scale fading into consideration \cite{Hieu,Yaxiong}. Concretely, the channel coefficient $h_{iu}^n$ at time slot $n$ is given as
\begin{align}
\label{eq:4}
h_{iu}^n = \sqrt {\psi_{iu}^n} {\tilde h}_{iu}^n,
\end{align}
where $\psi_{iu}^n$ and ${\tilde h}_{iu}^n$ denotes the large-scale fading and small-scale fading during time slot $n$, respectively. Specifically, $\psi_{iu}[n]$ can be written as
\begin{align}
\label{eq:5}
\psi_{iu}^n = \omega_0 (d_{iu}^n)^{-\alpha},
\end{align}
where $\omega_0$ represents the reference channel gain at $d_{iu} = 1$ meter, and $\alpha$ denotes the path loss exponent. The small-scale fading ${\tilde h}_{iu}^n$ with $\mathbb{E} \left[ | {\tilde h}_{iu}^n  |^2 \right]=1$, can be modeled as
\begin{align}
\label{eq:6}
{\tilde h}_{iu}^n = \sqrt{\frac{K}{1+K}} \bar{h}_{iu}^n + \sqrt{\frac{1}{1+K}} \hat{h}_{iu}^n,
\end{align}
where $\bar{h}_{iu}^n$ accounts for deterministic LoS, $\hat{h}_{iu}^n$ denotes the NLoS component, and $K$ is the Rician factor.


\subsection{Energy Harvesting and Energy Consumption Constraints} 
Due to the limited energy storage at the UB, EH becomes a promising
solution in prolonging the lifetime of the UB. We design a dynamic time-splitting mechanism wherein the UB communication can be divided into two dynamic phases within a time slot. Specifically, a fraction $\tau_n$ and $(1-\tau_n)$ of duration $\delta_t$ are used for backscattering signal and EH at the UB, respectively. In the second phase of $(1-\tau_n)\delta_t$, the harvested energy expression at the UB at time slot $n$ is given by \cite{lyu2017optimal}.
\begin{subnumcases} {\label{eq:7} E_h^n =}
E_{\rm L}^n \triangleq \mu (1-\tau_n) \delta_t P_s \mathbb{E}[h^n_{su}], 
\label{eq:7a} \\
E_{\rm NL}^n \triangleq (1-\tau_n)\delta_t \frac{\Xi}{1-\phi} \big(\frac{1}{1+e^{-\beta P_s \mathbb{E}[h^n_{su}]+\beta\nu}} -\phi \big) \label{eq:7b}
\end{subnumcases} 
where $P_s$ is the transmit power at the source, $E_{\rm L}^n$ and $E_{\rm NL}^n$ are the harvested energy corresponding to linear and sigmoidal function based non-linear models, respectively, $\mu$ denotes the energy harvesting efficiency corresponding to the LEH model and $\tau_n$ represents the DTS ratio at time slot $n$. More specifically, $\tau_n=1$ means that all the signal is backscattered to the receiver during time slot $n$  and $\tau_n=0$ indicates that all the signal is used for EH. Furthermore, $\phi$ is defined as $\phi \triangleq \frac{1}{1+e^{\beta \nu}}$, wherein $\beta$ and $\nu$ are constant values with regards to circuit specifications such as diode turn-on voltage, capacitance, and resistance. $\Xi$ is the maximum harvested power at the UB when the EH circuit is saturated \cite{Sumit2020Backcom}.

The energy consumption due to BackCom during time slot $n$ is represented as $\tau_n \delta_t P_c$, where $P_c$ is the circuit power of the UB during backscatter period  \cite{lyu2017optimal}. We then have the following energy constraint
\begin{align}
\label{eq:8}
\sum\limits_{i=1}^n  \tau_n \delta_t P_c \le \sum\limits_{i=1}^n  E_h^n ,
\end{align}
The constraint \eqref{eq:8} guarantees that the total UAV's energy consumption should be less than or equal to the summation of harvested energy of the UB until time slot $n \in {\cal N}$.

By substituting \eqref{eq:7} into \eqref{eq:8}, we have
\begin{subnumcases} {\label{eq:9} \sum\limits_{n \in {\cal N}} \tau_n \delta_t P_c \le} \sum\limits_{n \in {\cal N}}  \frac{\mu (1-\tau_n) \delta_t \omega_0 P_s} {\big(H^2+{{\left\| {{\bq}_{n} - {\bw}_{s}} \right\|}^2}\big)^{\alpha/2}}, 
\label{eq:9a} \\
\sum\limits_{n \in {\cal N}} (1-\tau_n)\delta_t \frac{\Xi}{1-\phi} \Bigg(\frac{1}{1+e^{\frac{-\beta P_s\omega_0}{(H^2+{{\left\| {{\bq}_{n} -{\bw}_{s}} \right\|}^2})^{\alpha/2}}+\beta\nu} } - \phi \Bigg), \label{eq:9b}
\end{subnumcases}
where $\mathbb{E}[h^n_{su}] = \frac{\omega_0}{\big(H^2+{{\left\| {\bq}_n -{\bw}_s \right\|}^2}\big)^{\alpha/2}}$.
\subsection{Caching Model}
We consider a general caching model at the UB, whereas the UB needs to retrieve the information from its cache to serve the destination. Specifically, the UB is able to store $0 \le \sigma \le 1$ parts of each file in its cache \cite{gautam2018cache} \footnote{This caching method is also known as probabilistic caching.}. Henceforth, $\sigma$ is considered as the caching coefficient in this paper. When the destination requests a file, a part $\sigma$ of this file is already stored in the UB's storage. Therefore, the source only needs to send the remainder of the required file to the UB before its transmission to the destination via backscatter. Moreover, the caching scheme adopted in this work can be considered as a lower bound method in comparison with the case when UB knew the content popularity.

\vspace{-0.2cm}
\subsection{UAV-enabled backscatter (UB) }
In this work, we consider a UB as a flying backscatter device to reflect the signal from source to destination. To avoid the co-channel interference on the uplink (UL) and downlink (DL), time-division duplexing (TDD) is utilized in this system \cite{hua2019}. Specifically, we consider DTS method to divide each time slot into two parts. In this context, $(1-\tau_n)\delta_t$ and $\tau_n\delta_t$ are the fraction of time for data transmission on the UL from $s \to u$ and the DL from $u \to d$, respectively, where $0 \le \tau_n \le 1$ denotes the DTS ratio at the time slot $n$.

Let us denote the symbol transmitted from the source during time slot $n$ by $x_s^n$ with unit power $\mathbb{E}[|x_s^n|^2]=1$. Then, the received signal at the UB during time slot $n$ is given by
\begin{align}
\label{eq:10}
y_u^n = \sqrt{P_s}h_{su}^n x_s^n,
\end{align}
Notably, the noise factor is ignored at the UB since the backscatter's circuit only includes the passive components without the active components such as amplifiers, analog-to-digital (ADC) converter, or oscillators \cite{qian2016,wang2016,xiao2019resource}. Let us denote $c_u^n$ as the backscatter information signal at time slot $n$, the transmitted signal of the UB is then given as \cite{xiao2019resource}
\begin{align}
\label{eq:11}
x_u^n = \sqrt{\eta_u^n P_s}h_{su}^n x_s^n c^n_u,
\end{align}
where $\eta_u^{n}$ represents the backscatter coefficient during time slot $n$, $c_u^{n}$ denotes the UB's own information with $\mathbb{E}[|c_s^n|^2]=1$ at time slot $n$. Note that $\eta_u^n$ can not reach 1 in practice due to material and circuit losses \cite{xiao2019resource}. Hence, we set a threshold for $\eta_u^n \le \eta_{\rm max}$, with $0 < \eta_{\rm max} < 1$. Moreover, the additional noise and signal processing delay are ignored in \eqref{eq:11} which are widely utilized in \cite{xiao2019resource,lyu2017optimal,kang2018}. Consequently, the received signal at destination during time slot $n$ is given as
\begin{align}
\label{eq:12}
y_d^{n} = h_{ud}^{n} x_u^n + n_d,
\end{align}
where $n_d\sim {\cal{CN}}(0,\sigma_d^2)$ denote the additive white Gaussian noise (AWGN) at the destination. By substituting \eqref{eq:11} into \eqref{eq:12}, we have
\begin{align}
\label{eq:13}
y_d^{n} = \sqrt{ \eta_u^{n} P_s}h_{su}^n h_{ud}^{n} x_s^n c_u^{n} + n_d, 
\end{align}
where $\sqrt{P_s}h_{sd}^n x_s^n$ is the transmitted signal from $u\to d$ during time slot $n$ and $n_d$ is the noise power at the destination which is an independent and identically distributed (i.i.d.) complex Gaussian random variable with zero mean and variance $\sigma_d^2$. Thus, the SNR at the destination are represented as
\begin{align}
	\label{eq:14}
\gamma_d^{n} = \frac{ \eta_u^{n} P_s |h_{su}^n|^2 |h_{ud}^{n}|^2 }{\sigma_d^2},
\end{align}

Then, the achievable rate (in bps) at the UB and the destination during time slot $n$ can be respectively calculated as
\begin{align}
\label{eq:15}
R_u^{n} =  B \log_2\big(1+\gamma_u^{n}\big),\\
\label{eq:16}
R_d^{n} =  B \log_2\big(1+\gamma_d^{n}\big),
\end{align}
where $B$ denotes the sytem bandwidth in hertz (Hz); $\gamma_u^{n} = P_s |h^n_{su}|^2 $, $B$ is the total bandwidth. Especially, the instantaneous channel state information (CSI) (i.e., $h_{su}^n$ and  $h_{ud}^{n}$) are random variables, thus the instantaneous rate is also a random variable. Thus, the approximated received rate of the UB and the destination are adopted, which can be expressed as \cite{hua2019}
\begin{align}
\label{eq:17}
\bar{R}_u^{n} &= B \mathbb{E}[ \log_2 \big(1+\gamma_u^{n}\big)],\\
\label{eq:18}
\bar{R}_d^{n} &= B \mathbb{E}[ \log_2 \big(1+\gamma_d^{n}\big)]).
\end{align}

As explicit, it is difficult to obtain the closed-form expression of $\bar{R}_u^{n}$ and $\bar{R}_d^{n}$, and hence the approximation functions for $\bar{R}_u^{n}$ and $\bar{R}_d^{n}$ are expressed as in the following lemma:
\vspace{-0.1cm}
\begin{lemma}\label{lemma:1}
	The approximation expressions of $\bar{R}_u^{n}$ and $\bar{R}_d^{n}$ are respectively given as 
    \begin{align}
	\label{eq:Lemma1_1}
	\bar{R}_u^{n} &= B \log_2 \Bigg(1+  \frac{ \omega_0 P_s }{\big({H^2} + {{\left\| {\bq}_{n} - {\bw}_{s} \right\|}^2}\big)^{\alpha/2} } \Bigg), \\
	\label{eq:Lemma1_2}
	\bar{R}_d^{n} &= B \log_2 \Bigg(1+  \frac{ \Theta \eta_u^{n} P_s}{\big({H^2} + {{\left\| {{\bq}_n - {\bw}_{s}} \right\|}^2}\big)^{\alpha/2} \big({H^2} + {{\left\| {{\bq}_n - {\bw}_{d}} \right\|}^2}\big)^{\alpha/2} } \Bigg),
	\end{align}
\end{lemma}
where $\Theta \triangleq \frac{e^{-E}\omega_0^2 }{\sigma_d^2}$.
\begin{proof}
See Appendix~A.
\end{proof}

\vspace{-0.3cm}
\subsection{Problem Formulation}
This section aims at maximizing the total data transmission from $u \to d$ by jointly optimizing the backscatter coefficient, DTS ratio, and UB trajectory with consideration of a linear EH model. Let us define $\bq \triangleq \{{\bq}_n, n \in {\cal N} \}$, $\boldsymbol \eta=\{\eta_u^{n}, n \in {\cal N} \}$, ${\boldsymbol \tau \triangleq \{ \tau_{n}, n \in {\cal N} \}}$. Then, the problem is mathematically formulated as follows 
\begin{IEEEeqnarray}{rCl}\label{eq:P_1}
	{\cal P}_1: &&\max_{\bq, {\boldsymbol \eta}, {\boldsymbol \tau} }~~ B \sum\limits_{n \in {\cal N}} \tau_{n} \delta_t \log_2 \Bigg(1+  \frac{ \Theta \eta_u^{n} P_s}{\big({H^2} + {{\left\| {{\bq}_n - {\bw}_{s}} \right\|}^2}\big)^{\alpha/2} \big({H^2} + {{\left\| {{\bq}_n - {\bw}_{d}} \right\|}^2}\big)^{\alpha/2} } \Bigg)  \IEEEyessubnumber \label{eq:P1:a}\\
	\mathtt{s.t.}~~
	&& B \sum\limits_{n \in {\cal N}} \tau_{n} \delta_t \log_2 \Bigg(1+  \frac{ \omega_0 P_s }{\big({H^2} + {{\left\| {{\bq}_n - {\bw}_{s}} \right\|}^2}\big)^{\alpha/2} } \Bigg) + \sigma S \notag \\ &&\ge B \textbf{ }\sum\limits_{n \in {\cal N}} \tau_{n} \delta_t \log_2 \Bigg(1+  \frac{ \Theta \eta_u^{n} P_s}{\big({H^2} + {{\left\| {{\bq}_n - {\bw}_{s}} \right\|}^2}\big)^{\alpha/2} \big({H^2} + {{\left\| {{\bq}_n - {\bw}_{d}} \right\|}^2}\big)^{\alpha/2} } \Bigg), \IEEEyessubnumber\label{eq:P1:b}\\
	&& B \sum\limits_{n \in {\cal N}} \tau_{n} \delta_t \log_2 \Bigg(1+  \frac{ \Theta \eta_u^{n} P_s}{\big({H^2} + {{\left\| {{\bq}_n - {\bw}_{s}} \right\|}^2}\big)^{\alpha/2} \big({H^2} + {{\left\| {{\bq}_n - {\bw}_{d}} \right\|}^2}\big)^{\alpha/2} } \Bigg) \ge S, \IEEEyessubnumber\label{eq:P1:c}\\
	&& \sum\limits_{i=1}^n  \tau_n \delta_t P_c  \le \sum\limits_{i=1}^n \frac{\mu (1-\tau_n) \delta_t \omega_0 P_s} {\big(H^2+{{\left\| {{\bq}_n - {\bw}_{s}} \right\|}^2}\big)^{\alpha/2}}  ,
	\IEEEyessubnumber\label{eq:P1:d}\\
	&& \left\| {\bq}_{n+1} - {\bq}_{n} \right\| \le \delta_d = V_{\rm max}{\delta _t},n=0,\dots,N-1, \IEEEyessubnumber\label{eq:P1:e}\\
	&& {\bq}_0 = {\bq}_{\rm I}, {\bq}_N={\bq}_{\rm F}, \IEEEyessubnumber\label{eq:P1:f}\\
	&& 0 \le \tau_{n} \le 1, n \in {\cal N},
	\IEEEyessubnumber\label{eq:P1:g}\\
	&&  0 \le \eta_u^{n} \le \eta_{\rm max}, \text{with}\; 0 < \eta_{\rm max} < 1,  n \in {\cal N}, \IEEEyessubnumber\label{eq:P1:h}
\end{IEEEeqnarray} 
where $S$ is the demanded data (in bits) by the destination; constraint \eqref{eq:P1:b} guarantees a non-empty caching at the UB; constraint \eqref{eq:P1:c} means that the total transmitted data on the DL from $u \to d$ should be larger than or equal to the demanded data of the destination; constraint \eqref{eq:P1:g} implies that the DTS ratio value must be less than or equal to 1 and constraint \eqref{eq:P1:h} signifies the limitation on the backscatter coefficient.

The problem ${\cal P}_1$ is a mixed integer non-linear program (MINLP), which is NP-hard. Specifically, the objective function, constraints \eqref{eq:P1:b}, \eqref{eq:P1:c}, \eqref{eq:P1:d} are non-convex. Thus, it is troublesome to find the direct solution of ${\cal P}_1$. In the succeeding section, we introduce an efficient method to solve it.

\section{Proposed Alternating Algorithm for Solving ${\cal P}_1$}
\label{sec:3}
To tackle the non-convexity of the problem ${\cal P}_1$, we first decompose ${\cal P}_1$ into three sub-problems, wherein we first target the optimization of DTS ratio for a given trajectory and backscatter coefficient, followed by the optimization of backscatter coefficient for a given trajectory and DTS ratio, and finally we perform the trajectory optimization for a given DTS ratio and backscatter coefficient. By employing the block coordinate descent (BCD) method \cite{hong2015unified}, we propose an efficient iterative algorithm wherein we alternately optimize three subproblems until the algorithm converges to a given threshold, $\epsilon >0$.

\subsection{Dynamic Time Splitting Ratio Optimization:} For any given UB trajectory $\bq$ and backscatter coefficient $\boldsymbol{\eta}$, the DTS ratio $\boldsymbol{\tau}$ can be obtained by solving the following optimization problem:
\begin{IEEEeqnarray}{rCl}\label{eq:P_1t}
	{\cal P}_1^{\boldsymbol{\tau}}: &&\max_{ {\boldsymbol \tau} }~~ \sum\limits_{n \in {\cal N}} \tau_{n} \delta_t \bar{R}_d^{n}  \IEEEyessubnumber \label{eq:P1t:a}\\
	\mathtt{s.t.}~~
	&& \sum\limits_{n \in {\cal N}} \tau_{n} \delta_t \bar{R}_u^{n} + \sigma S \notag \ge \sum\limits_{n \in {\cal N}} \tau_{n} \delta_t \bar{R}_d^{n}, \IEEEyessubnumber\label{eq:P1t:b}\\
	&& \sum\limits_{n \in {\cal N}} \tau_{n} \delta_t \bar{R}_d^{n} \ge S, \IEEEyessubnumber\label{eq:P1t:c}\\
	&& \sum\limits_{i=1}^n  \tau_n \delta_t P_c \le \sum\limits_{i=1}^n \frac{\mu (1-\tau_n) \delta_t \omega_0 P_s} {\big(H^2+{{\left\| {{\bq}_n - {\bw}_{s}} \right\|}^2}\big)^{\alpha/2}} ,
	\IEEEyessubnumber\label{eq:P1t:d}\\
	&& 0 \le \tau_{n} \le 1, n \in {\cal N}.
	\IEEEyessubnumber\label{eq:P1t:e}
\end{IEEEeqnarray} 

It is clear that ${\cal P}_1^{\boldsymbol{\tau}}$ is a linear optimization problem, and hence is convex. Moreover, it is easy to verify that the Slater's condition holds for ${\cal P}_1^{\boldsymbol{\tau}}$ and thus the KKT conditions are sufficient for optimality
\cite[Section 5.5]{Boy}. Then, the Lagrangian function corresponding to problem ${\cal P}_1^{\boldsymbol{\tau}}$ is expressed as
\begin{align}
\label{eq:B1}
{\cal{L}} ({\boldsymbol{\tau}}, \lambda_1, \lambda_2, \lambda_3, \lambda_4) \triangleq  F(\tau) +  \lambda_1 G(\boldsymbol{\tau}) + \lambda_2 H(\boldsymbol{\tau}) + \lambda_3 I(\boldsymbol{\tau}) + \lambda_4 J(\boldsymbol{\tau}),
\end{align}

with 
\begin{align}
\vspace{-0.1cm}
\label{eq:B2}
F(\boldsymbol{\tau}) &\triangleq \sum\limits_{n \in {\cal N}} \tau_{n} \delta_t \bar{R}_d^{n} \\
\label{eq:B3}
\vspace{-0.1cm}
G({\boldsymbol{\tau}}) &\triangleq \Big( \sum\limits_{n \in {\cal N}} \tau_{n} \delta_t \bar{R}_u^{n} + \sigma S - \sum\limits_{n \in {\cal N}} \tau_{n} \delta_t \bar{R}_d^{n} \Big) \ge 0 \\
\label{eq:B4}
H({\boldsymbol{\tau}}) &\triangleq \sum\limits_{n \in {\cal N}} \tau_{n} \delta_t \bar{R}_d^{n} - S \ge 0, \\
\label{eq:B5} I({\boldsymbol{\tau}}) &\triangleq  \sum\limits_{i=1}^n \chi_1(1-\tau_n) - \sum\limits_{i=1}^n \tau_n \delta_t P_c \ge 0, \\
\label{eq:B6}
J({\boldsymbol{\tau}}) &\triangleq 1 - \tau_n \ge 0 .
\end{align}
where $\lambda_1,$  $\lambda_2$, $\lambda_3,$ $\lambda_4$ is the Lagrangian dual variables; $\chi_1 \triangleq \frac{\mu  \delta_t \omega_0 P_s} {\big(H^2+{{\left\| {{\bq}_n - {\bw}_{s}} \right\|}^2}\big)^{\alpha/2}}$.

The stationarity condition is given as
\begin{align}
\label{eq:B7}
\frac{\partial {\cal L}({\boldsymbol{\tau}}, \lambda_1, \lambda_2, \lambda_3, \lambda_4)}{\partial {\boldsymbol{\tau}}} &= \sum\limits_{n \in {\cal N}}  \delta_t \bar{R}_d^{n} + \lambda_1 \Big( \sum\limits_{n \in {\cal N}} \delta_t \bar{R}_u^{n}  - \sum\limits_{n \in {\cal N}}  \delta_t \bar{R}_d^{n} \Big) + \lambda_2 \sum\limits_{n \in {\cal N}} \delta_t \bar{R}_d^{n} \notag \\ & - \lambda_3 \Big(\sum\limits_{n \in {\cal N}}\chi_1 + \sum\limits_{n \in {\cal N}} \delta_t P_c\Big) - \lambda_4 = 0
\end{align}

The conditions for primal feasibility are given as \eqref{eq:B3}, \eqref{eq:B4}, \eqref{eq:B5}, and \eqref{eq:B6}. Then, the complementary slackness conditions can be expressed as follows
\begin{align}
\label{eq:B8}
\lambda_1 G(\boldsymbol{\tau}) = 0,\\
\label{eq:B9}
\lambda_2 H(\boldsymbol{\tau}) = 0, \\
\label{eq:B10}
\lambda_3 I(\boldsymbol{\tau}) = 0, \\
\label{eq:B11}
\lambda_4 J(\boldsymbol{\tau}) = 0.
\end{align}

Furthermore, the dual feasibility conditions should hold $\lambda_1,\lambda_2,\lambda_3,\lambda_4 \ge 0$. The solution is then postulated in the following
theorem.
\begin{theorem}
\label{theorem:1}
The optimal value $\{\tau_n^\star\}$ to problem ${\cal P}_{1}^{\boldsymbol{\tau}}$ can be expressed as
\begin{subnumcases} {\label{eq:23} \tau_n^\star =}
\Bigg[\frac{\sigma S}{N \delta_t (\bar{R}^n_d - \bar{R}^n_u)}\Bigg]^1_0, \hfill \; {\rm iff} \bar{R}^n_d > \bar{R}^n_u, \forall n \in {\cal N} 
\label{eq:23a} \\
\Bigg[\frac{\chi_1 } { \chi_1 +  \delta_t P_c}\Bigg]^1_0. \label{eq:23b}
\end{subnumcases}
From Eq. \eqref{eq:23}, we report two possible solutions of $\{\tau_n^\star\}$. In order to reach to an optimal outcome, we select the best solution that maximizes the objective function in ${\cal P}_1^\tau$. Particularly, it is easy to verify that the optimal value of $\{\tau_n^\star\}$ from $\eqref{eq:23b}$ is guaranteed to be inside the feasible set, i.e., $0 \le \tau_n \le 1$.

\begin{proof}
	See Appendix~B.
\end{proof}

\end{theorem}

\subsection{Backscatter Coefficient Optimization:} For any given trajectory $\bq$ and DTS ratio $\boldsymbol{\tau}$, the backscatter coefficient $\boldsymbol{\eta}$ can be achieved by solving the following optimization problem:
\begin{IEEEeqnarray}{rCl}\label{eq:Pe_1}
	{\cal P}_1^{\boldsymbol{\eta}}: &&\max_{ {\boldsymbol \eta} }~~ B \sum\limits_{n \in {\cal N}} \tau_{n} \delta_t \log_2 \Big(1+  \varphi_1 \eta_u^n \Big)  \IEEEyessubnumber \label{eq:P1e:a}\\
	\mathtt{s.t.}~~
	&& B \sum\limits_{n \in {\cal N}} \tau_{n} \delta_t \log_2 \Big(1+  \varphi_2 \Big) + \sigma S \ge B \sum\limits_{n \in {\cal N}} \tau_{n} \delta_t \log_2 \Big(1+  \varphi_1 \eta_u^n \Big), \IEEEyessubnumber\label{eq:P1e:b}\\
	&& B \sum\limits_{n \in {\cal N}} \tau_{n} \delta_t \log_2 \Big(1+  \varphi_1 \eta_u^n \Big) \ge S, \IEEEyessubnumber\label{eq:P1e:c}\\
	&&  0 \le \eta_u^{n} \le \eta_{\rm max}, \; \text{with}\; 0 < \eta_{\rm max} < 1,  n \in {\cal N}, \IEEEyessubnumber\label{eq:P1e:d}
\end{IEEEeqnarray} 
where $\varphi_1 \triangleq \frac{ \Theta P_s}{\big({H^2} + {{\left\| {{\bq}_n - {\bw}_{s}} \right\|}^2}\big)^{\alpha/2} \big({H^2} + {{\left\| {{\bq}_n - {\bw}_{d}} \right\|}^2}\big)^{\alpha/2} }$, $\varphi_2 \triangleq \frac{ \omega_0 P_s }{\big({H^2} + {{\left\| {{\bq}_n - {\bw}_{s}} \right\|}^2}\big)^{\alpha/2} }$. 

\vspace{0.2cm}
It is noteworthy that ${\cal P}_1^{\boldsymbol{\eta}}$
is a much simpler form as compared with ${\cal P}_1$, but it is still inviable to obtain the direct solution due to the non-convexity of the constraint \eqref{eq:P1e:b}. To convexify \eqref{eq:P1e:b}, an efficiently approximate solution is obtained based on the Successive Convex Approximation (SCA) method. To proceed, we define $\eta_u^{n,j}$ as the given backscatter coefficient at $j$-th iteration. Based on the concavity of $\log_2 (1 + x)$ w.r.t. $x$, we have
\begin{align}
\label{eq:36}
B \sum\limits_{n \in {\cal N}} \tau_{n} \delta_t \log_2 \Big(1+  \varphi_1 \eta_u^n \Big) \le B \sum\limits_{n \in {\cal N}} \tau_{n} \delta_t \log_2 \Big(1+  \varphi_1 \eta_u^{n,j} \Big) + \frac{B \sum\limits_{n \in {\cal N}} \tau_{n} \delta_t \varphi_1 \big(\eta_u^n-\eta_u^{n,j} \big)}{\ln 2  \big(1+\varphi_1 \eta_u^{n,j} \big)} \triangleq \tilde{R}_u^{j},
\end{align}

Consequently, ${\cal P}_1^{\boldsymbol{\eta}}$ can be re-written as
\begin{IEEEeqnarray}{rCl}\label{eq:Pe2_1}
	{\cal P}_{1.1}^{\boldsymbol{\eta}}: &&\max_{ {\boldsymbol \eta} }~~ B \sum\limits_{n \in {\cal N}} \tau_{n} \delta_t \log_2 \Big(1+  \varphi_1 \eta_u^n \Big)  \IEEEyessubnumber \label{eq:P1e2:a}\\
	\mathtt{s.t.}~~
	&& B \sum\limits_{n \in {\cal N}} \tau_{n} \delta_t \log_2 \Big(1+  \varphi_2 \Big) + \sigma S \ge \tilde{R}_u^{j}, \IEEEyessubnumber\label{eq:P1e2:b}\\
	&& \eqref{eq:P1e:c},\eqref{eq:P1e:d}, \IEEEyessubnumber\label{eq:P1e2:c}
\end{IEEEeqnarray}
which is convex. Furthermore, it is easy to verify that the Slater's conditions holds for ${\cal P}_{1.1}^{\boldsymbol{\eta}}$ and hence the solution can be obtained by solving the dual Lagrangian problem which significantly reduces the computational complexity \cite[Section 5.5]{Boy}. 

The corresponding Lagrangian for $	{\cal P}_{1.1}^{\boldsymbol{\eta}}$ is given by
\begin{align}
\label{eq:38}
{\cal L}_1(\boldsymbol{\eta}) \triangleq F_1(\boldsymbol{\eta}) + \lambda_5 G_1(\boldsymbol{\eta}) + \lambda_6 H_1(\boldsymbol{\eta}),
\end{align}
with
\begin{align}
\label{eq:39}
F_1(\boldsymbol{\eta}) &\triangleq B \sum\limits_{n \in {\cal N}} \tau_{n} \delta_t \log_2 \Big(1+  \varphi_1 \eta_u^n \Big), \\
\label{eq:40}
G_1(\boldsymbol{\eta}) &\triangleq \varphi_3 - \varphi_4 - \varphi_5 \big(\eta_u^n-\eta_u^{n,j} \big) \ge 0 \\
\label{eq:41}
H_1(\boldsymbol{\eta}) &\triangleq B \sum\limits_{n \in {\cal N}} \tau_{n} \delta_t \log_2 \Big(1+  \varphi_1 \eta_u^n \Big) - S \ge 0.
\end{align}
where $\varphi_3 \triangleq B \sum\limits_{n \in {\cal N}} \tau_{n} \delta_t \log_2 \Big(1+  \varphi_2 \Big) + \sigma S,$ $\varphi_4 \triangleq B \sum\limits_{n \in {\cal N}} \tau_{n} \delta_t \log_2 \Big(1+  \varphi_1 \eta_u^{n,j} \Big)$, $\varphi_5 \triangleq \frac{B \sum\limits_{n \in {\cal N}} \tau_{n} \delta_t \varphi_1 }{\ln 2  \big(1+\varphi_1 \eta_u^{n,j} \big)}$. 

For (local) optimality, the stationarity condition must hold
\begin{align}
\label{eq:42}
\frac{\partial {\cal L}_1({\boldsymbol{\eta}}, \lambda_5, \lambda_6)}{\partial {\boldsymbol{\eta}}} = \frac{B \sum\limits_{n \in {\cal N}} \tau_{n} \delta_t \varphi_1 }{\ln 2 \Big(1+  \varphi_1 \eta_u^n \Big)} - \lambda_5 \varphi_5 + \frac{\lambda_6 B \sum\limits_{n \in {\cal N}} \tau_{n} \delta_t \varphi_1 }{\ln 2 \Big(1+  \varphi_1 \eta_u^n \Big)} = 0.
\end{align}

The primal feasibility conditions are expressed as in \eqref{eq:40}, \eqref{eq:41}, and \eqref{eq:42}. Then, the complementary slackness conditions are given as
\begin{align}
\label{eq:43}
\lambda_5 G_1(\boldsymbol{\eta}) = 0,\\
\label{eq:44}
\lambda_6 H_1(\boldsymbol{\eta}) = 0.
\end{align}
The solution is then postulated in the following
theorem.
\begin{theorem}
	\label{theorem:2}
	The optimal value $\{\eta_u^{n,\star}\}$ to problem ${\cal P}_{1.1}^{\boldsymbol{\eta}}$ can be expressed as
	\begin{align}
	\label{eq:45}
	\eta_u^{n,\star} = \bigg[\frac{\varphi_3 - \varphi_4 + \varphi_5 \eta_u^{n,j}}{\varphi_5}\bigg]^{\eta_{\rm max}}_0.
	\end{align}
	\begin{proof}
		See Appendix~C.
	\end{proof}
\end{theorem}


\subsection{Trajectory Optimization:}
\label{subsec:c}
For given values of $\boldsymbol{\eta}$ and $\boldsymbol{\tau}$, the UAV trajectory $\bq$ can be achieved by solving the following problem
\begin{IEEEeqnarray}{rCl}\label{eq:Pq_1}
	\vspace{-0.05cm}
	{\cal P}_1^{\bq}: &&\max_{\bq }~~ B \sum\limits_{n \in {\cal N}} \tau_{n} \delta_t \log_2 \Bigg(1+  \frac{ \Theta \eta_u^{n} P_s}{\big({H^2} + {{\left\| {{\bq}_n - {\bw}_{s}} \right\|}^2}\big)^{\alpha/2} \big({H^2} + {{\left\| {{\bq}_n - {\bw}_{d}} \right\|}^2}\big)^{\alpha/2} } \Bigg)  \IEEEyessubnumber \label{eq:P1q:a}\\
	\mathtt{s.t.}~~
	\vspace{-0.1cm}
	&& B \sum\limits_{n \in {\cal N}} \tau_{n} \delta_t \log_2 \Bigg(1+  \frac{ \omega_0 P_s^{n} }{\big({H^2} + {{\left\| {{\bq}_n - {\bw}_{s}} \right\|}^2}\big)^{\alpha/2} } \Bigg) + \sigma S \notag \\ \vspace{-0.05cm} &&\ge B \sum\limits_{n \in {\cal N}} \tau_{n} \delta_t \log_2 \Bigg(1+  \frac{ \Theta \eta_u^{n} P_s}{\big({H^2} + {{\left\| {{\bq}_n - {\bw}_{s}} \right\|}^2}\big)^{\alpha/2} \big({H^2} + {{\left\| {{\bq}_n - {\bw}_{d}} \right\|}^2}\big)^{\alpha/2} } \Bigg), \IEEEyessubnumber\label{eq:P1q:b}\\ \vspace{-0.05cm}
	&& B \sum\limits_{n \in {\cal N}} \tau_{n} \delta_t \log_2 \Bigg(1+  \frac{ \Theta \eta_u^{n} P_s}{\big({H^2} + {{\left\| {{\bq}_n - {\bw}_{s}} \right\|}^2}\big)^{\alpha/2} \big({H^2} + {{\left\| {{\bq}_n - {\bw}_{d}} \right\|}^2}\big)^{\alpha/2} } \Bigg) \ge S, \IEEEyessubnumber\label{eq:P1q:c}\\
	&& \sum\limits_{i=1}^n  \tau_n \delta_t P_c  \le  \sum\limits_{i=1}^n  \frac{\mu (1-\tau_n) \delta_t \omega_0 P_s} {\big(H^2+{{\left\| {{\bq}_n - {\bw}_{s}} \right\|}^2}\big)^{\alpha/2}},
	\IEEEyessubnumber\label{eq:P1q:d}\\
	&& \eqref{eq:P1:e}, \eqref{eq:P1:f}, \IEEEyessubnumber\label{eq:P1q:e}
\end{IEEEeqnarray}

The problem ${\cal P}_1^{\bq}$ is a MINLP, which is difficult to  efficiently solve by utilizing standard optimization methods. To make ${\cal P}_1^{\bq}$ more tractable, we firstly introduce slack variables $z_1^n$ and $z_2^n$ such that $\big({H^2} + {{\left\| {{\bq}_n - {\bw}_{s}} \right\|}^2}\big) \le (z_1^n)^{2/\alpha}$ and $\big({H^2} + {{\left\| {{\bq}_n - {\bw}_{d}} \right\|}^2}\big) \le (z_2^n)^{2/\alpha}$, respectively. Let us denote ${\bold z} \triangleq \{z_1^n, z_2^n, n\in {\cal N}\}$, by which the problem ${\cal P}_1^{\bq}$ is rewritten as
\begin{IEEEeqnarray}{rCl}\label{eq:P1q1_1}
	{\cal P}_{1.1}^{\bq}: &&\max_{\bq, {\bold z} }~~ B \sum\limits_{n \in {\cal N}} \tau_{n} \delta_t \log_2 \Bigg(1+  \frac{ \Theta \eta_u^{n} P_s}{z_1^n z_2^n} \Bigg)  \IEEEyessubnumber \label{eq:P1q1:a}\\
	\mathtt{s.t.}~~
	&& \big({H^2} + {{\left\| {{\bq}_n - {\bw}_{s}} \right\|}^2}\big) \le (z_1^n)^{2/\alpha}, \IEEEyessubnumber\label{eq:P1q1:b}\\
	&& \big({H^2} + {{\left\| {{\bq}_n - {\bw}_{d}} \right\|}^2}\big) \le (z_2^n)^{2/\alpha}, \IEEEyessubnumber\label{eq:P1q1:c}\\ 
	&& B \sum\limits_{n \in {\cal N}} \tau_{n} \delta_t \log_2 \Bigg(1+  \frac{ \omega_0 P_s^{n} }{z_1^n } \Bigg) + \sigma S \ge B \sum\limits_{n \in {\cal N}} \tau_{n} \delta_t \log_2 \Bigg(1+  \frac{ \Theta \eta_u^{n} P_s}{z_1^n z_2^n} \Bigg), \IEEEyessubnumber\label{eq:P1q1:d}\\
	&& B \sum\limits_{n \in {\cal N}} \tau_{n} \delta_t \log_2 \Bigg(1+  \frac{ \Theta \eta_u^{n} P_s}{z_1^n z_2^n } \Bigg) \ge S, \IEEEyessubnumber\label{eq:P1q1:e}\\
	&& \sum\limits_{i=1}^n  \tau_n \delta_t P_c \le  \sum\limits_{i=1}^n   \frac{\mu (1-\tau_n) \delta_t \omega_0 P_s} {z_1^n },
	\IEEEyessubnumber\label{eq:P1q1:f}\\
	&& \eqref{eq:P1:e}, \eqref{eq:P1:f}, \IEEEyessubnumber\label{eq:P1q1:g}
\end{IEEEeqnarray}

Note that the problem ${\cal P}_{1.1}^{\bq}$ is simpler than ${\cal P}_{1}^{\bq}$, but it is still difficult to be directly solved. This is because the objective function is convex and the non-convexity of constraints \eqref{eq:P1q1:d}, \eqref{eq:P1q1:e}, and \eqref{eq:P1q1:f}. In the following, we transform ${\cal P}_{1.1}^{\bq}$ into a convex form by introducing the following lemmas:
\begin{lemma}
\label{lemma:2}
For any given $z_1^{n,j}$ and $z_2^{n,j}$ at $j$-th iteration, $\log_2 \Big(1+  \frac{ \omega_0 P_s^{n} }{z_1^n } \Big)$ and $\log_2 \Big(1+  \frac{ \Theta \eta_u^{n} P_s}{z_1^n z_2^n} \Big)$ are respectively lower bounded by
\begin{align}
\label{eq:48}
\log_2 \bigg(1+  \frac{ \omega_0 P_s }{z_1^n } \bigg) &\ge \log_2 \bigg(1+  \frac{ \omega_0 P_s }{z_1^{n,j} } \bigg) - \frac{\omega_0 P_s(z_1^n-z_1^{n,j})}{ z_1^{n,j} (z_1^{n,j} + \omega_0 P_s)\ln 2} \triangleq \Theta_1 ,\\ \label{eq:49}
\log_2 \bigg(1+  \frac{ \Theta \eta_u^{n} P_s}{z_1^n z_2^n} \bigg) &\ge \log_2 \Big(1+  \frac{ \Theta \eta_u^{n} P_s}{z_1^{n,j} z_2^{n,j}} \Big) - \frac{\Theta \eta_u^{n} P_s (z_1^n-z_1^{n,j}) }{ z_1^{n,j}(z_1^{n,j} z_2^{n,j}+\Theta \eta_u^{n} P_s)\ln 2} \notag \\ &- \frac{\Theta \eta_u^{n} P_s(z_2^n-z_2^{n,j})}{z_2^{n,j}(z_1^{n,j} z_2^{n,j}+\Theta \eta_u^{n} P_s)\ln 2} \triangleq \Theta_2.
\end{align}
\begin{proof}
It is observed that $\log_2 \big(1+1/x\big)$ and $\log_2 \big(1+1/{xy}\big)$ are convex functions, with $x>0$ and $y>0$. Then, we adopt the first-order Taylor approximation to respectively approximate above convex functions at any given feasible points $x^j,y^j$ as
\begin{align}
\log_2 \big(1+\frac{A_1}{x}\big) &\ge \log_2 \big(1+\frac{A_1}{x^j}\big) - \frac{A_1}{x^j (x^j+A_1)\ln 2 } (x-x^j),\\
\log_2 \big(1+\frac{A_2}{xy}\big) &\ge \log_2 \big(1+\frac{A_2}{x^jy^j}\big) - \frac{A_2(x-x^j) }{ x^j(x^jy^j+A_2)\ln 2} - \frac{A_2(y-y^j)}{ y^j(x^jy^j+A_2)\ln 2} .
\end{align}

By applying $A_1 = \omega P_s$, $x=z_1^n$, $y=z_2^n$, and $A_2 = \Theta \eta_u^n P_s$, then the Lemma \ref{lemma:2} is proved.
\end{proof}
\end{lemma}

\begin{lemma}
\label{lemma:3}
For any given $z_1^{n,j}$ at the $j$-th iteration, the lower bound of $1/z_1^n$ can be expressed as
\begin{align}
\label{eq:52}
\frac{1}{z_1^n} \ge \frac{1}{z_1^{n,j}} - \frac{1}{(z_1^{n,j})^2} (z_1^n - z_1^{n,j}) \triangleq \tilde{z}_1^n.
\end{align}
\end{lemma}

Bearing all the above discussions in mind, we solve the following approximate convex
problem at the $j$-th iteration:
\begin{IEEEeqnarray}{rCl}\label{eq:P1q2_1}
	{\cal P}_{1.2}^{\bq}: &&\max_{\bq, {\bold z} }~~ B \sum\limits_{n \in {\cal N}} \tau_{n} \delta_t \Theta_2 \IEEEyessubnumber \label{eq:P1q2:a}\\
	\mathtt{s.t.}~~
	&& \eqref{eq:P1:e}, \eqref{eq:P1:f}, \eqref{eq:P1q1:b}, \eqref{eq:P1q1:c}, \IEEEyessubnumber\label{eq:P1q2:b} \\
	&& B \sum\limits_{n \in {\cal N}} \tau_{n} \delta_t \Theta_1 + \sigma S \ge B \sum\limits_{n \in {\cal N}} \tau_{n} \delta_t \Theta_2, \IEEEyessubnumber\label{eq:P1q2:c}\\
	&& B \sum\limits_{n \in {\cal N}} \tau_{n} \delta_t \Theta_2 \ge S, \IEEEyessubnumber\label{eq:P1q2:d}\\
	&& \sum\limits_{i=1}^n \tau_n \delta_t P_c \le \sum\limits_{i=1}^n  \mu (1-\tau_n) \delta_t \omega_0 P_s \tilde{z}_1^n .
	\IEEEyessubnumber\label{eq:P1q2:e}
\end{IEEEeqnarray}

Since the objective function and all constraints pf ${\cal P}_{1.1}^{\bq}$ are convex, thus it can be directly solved by applying standard optimization methods \cite{Boy}. To this end, we propose an iterative algorithm based on the solutions of three sub-problems. The alternating algorithm is summarized as in Algorithm \ref{Alg1}. 
\vspace{-0.1cm}
\begin{algorithm}[h]
	\begin{algorithmic}[1]
		\label{Alg1}
		\protect\caption{Proposed SCA-based Iterative Algorithm to Solve ${\cal P}_1$}
		\label{alg_1}
		\global\long\def\algorithmicrequire{\textbf{Initialization:}}
		\REQUIRE  Set $j:=0$ and initialize ${\boldsymbol{\eta}}^j$ and ${\bq}^j$.
		\vspace{-0.2cm}
		\REPEAT
		\vspace{-0.2cm}
		\STATE Solve ${\cal P}_1^{\boldsymbol{\tau}}$ for given $\{{\boldsymbol{\eta}}^j, {\bq}^j\}$ and denote the optimal solution as $\boldsymbol{\tau}^{j+1}$.
		\vspace{-0.2cm}
		\STATE Solve ${\cal P}_1^{\boldsymbol{\eta}}$ for given $\{{\boldsymbol{\tau}}^{j+1}, {\bq}^j\}$ and denote the optimal solution as $\boldsymbol{\eta}^{j+1}$.
		\vspace{-0.2cm}
		\STATE Solve ${\cal P}_1^{\bq}$ for given $ \{{\boldsymbol{\eta}}^{j+1}, {\boldsymbol{\tau}}^{j+1}\}  $ and denote the optimal solution as $\boldsymbol{\bq}^{j+1}$.
		\vspace{-0.2cm}
		\STATE Set $j:=j+1.$
		\vspace{-0.2cm}
		\UNTIL Convergence \\
\end{algorithmic} \end{algorithm}

\vspace{-0.1cm}
\subsection{Convergence and Complexity Analysis}
\subsubsection{Convergence Analysis}
\begin{proposition}
	\label{proposition_1}
	The proposed Algorithm \ref{Alg1} provides a solution that converges to at least a locally optimal solution.
\end{proposition}

\begin{proof}
Let us define $\Pi({\boldsymbol{\tau}}^j, {\boldsymbol{\eta}}^j, {\bq}^j)$, $\Pi^{\boldsymbol{\eta}}_{\rm lb}({\boldsymbol{\tau}}^j, {\boldsymbol{\eta}}^j, {\bq}^j)$, and $\Pi^{\bq}_{\rm lb}({\boldsymbol{\tau}}^j, {\boldsymbol{\eta}}^j, {\bq}^j)$ as the objective values of ${\cal P}_1$, ${\cal P}_{1.1}^{\boldsymbol{\eta}}$, and ${\cal P}_{1.2}^{\bq}$ at the $j$-th iteration. In the $(j+1)$-th iteration, at line 2 of Algorithm \ref{Alg1}, we have
\begin{align}
\label{eq:54}
\Pi({\boldsymbol{\tau}}^j, {\boldsymbol{\eta}}^j, {\bq}^j) \ile \Pi({\boldsymbol{\tau}}^{j+1}, {\boldsymbol{\eta}}^j, {\bq}^j).
\end{align}
The inequality $(i)$ holds since $\boldsymbol{\tau}^{j+1}$ is a optimal solution of ${\cal P}_1^{\boldsymbol{\tau}}$. Then, at line 3 of Algorithm \ref{Alg1}, we have
\begin{align}
	\label{eq:55}
	\Pi({\boldsymbol{\tau}}^{j+1}, {\boldsymbol{\eta}}^j, {\bq}^j) \iieq \Pi^{\boldsymbol{\eta}}_{\rm lb}({\boldsymbol{\tau}}^{j+1}, {\boldsymbol{\eta}}^j, {\bq}^j) \iiile  \Pi^{\boldsymbol{\eta}}_{\rm lb}({\boldsymbol{\tau}}^{j+1}, {\boldsymbol{\eta}}^{j+1}, {\bq}^j) \iiiile  \Pi({\boldsymbol{\tau}}^{j+1}, {\boldsymbol{\eta}}^{j+1}, {\bq}^j).	
\end{align}
The equality $(i2)$ holds since the first-order Taylor approximation at given point ${\boldsymbol{\eta}}^j$ is tight as in Eq. \eqref{eq:36}. Moreover, the inequality $(i3)$ holds since ${\boldsymbol{\eta}}^{j+1}$ is a optimal solution of ${\cal P}_{1.1}^{\boldsymbol{\eta}}$. Then, the inequality $(i4)$ holds since the objective value of ${\cal P}_{1.1}^{\boldsymbol{\eta}}$ is a lower bound to that of ${\cal P}_{1}^{\boldsymbol{\eta}}$ at given point ${\boldsymbol{\eta}}^{j+1}$. At line 4, we have
\begin{align}
	\label{eq:57}
	\Pi({\boldsymbol{\tau}}^{j+1}, {\boldsymbol{\eta}}^{j+1}, {\bq}^j) \ifiveeq \Pi^{\bq}_{\rm lb}({\boldsymbol{\tau}}^{j+1}, {\boldsymbol{\eta}}^{j+1}, {\bq}^j) \isixle \Pi^{\bq}_{\rm lb}({\boldsymbol{\tau}}^{j+1}, {\boldsymbol{\eta}}^{j+1}, {\bq}^{j+1}) \isevenle \Pi({\boldsymbol{\tau}}^{j+1}, {\boldsymbol{\eta}}^{j+1}, {\bq}^{j+1}).	
\end{align}
The equality $(i5)$ holds since the first-order Taylor approximation as in \eqref{eq:48}, \eqref{eq:49}, and \eqref{eq:52} are tight at given point ${\bq}^j$, and the inequality $(i6)$ holds since ${\bq}^{j+1}$ is a optimal solution of ${\cal P}_{1.2}^{\bq}$. Furthermore, the inequality $(i7)$ holds since the optimal value of ${\cal P}_{1.2}^{\bq}$ is a lower bound of ${\cal P}_{1}^{\bq}$ at given ${\bq}^{j+1}$. From \eqref{eq:57} and \eqref{eq:54}, we have $\Pi({\boldsymbol{\tau}}^j, {\boldsymbol{\eta}}^j, {\bq}^j) \le \Pi({\boldsymbol{\tau}}^{j+1}, {\boldsymbol{\eta}}^{j+1}, {\bq}^{j+1})$ which proves that the objective value of ${\cal P}_{1}$ is non-decreasing over the iterations. Moreover, the objective value of ${\cal P}_{1}$ is restricted by an upper bound value due to the limited total traveling time $T$, transmit power $P_s$, and maximum value of $\eta_{\max}$. Thus, the convergence of Algorithm \ref{Alg1} is assured.
\end{proof}

\subsubsection{Complexity Analysis} We provide the worst-case complexity analysis for Algorithm \ref{Alg1}. Since the problem ${\cal P}_1^{\boldsymbol{\tau}}$ and ${\cal P}_1^{\boldsymbol{\eta}}$ can be solved by using the proposed closed-form expressions, thus the complexity is mainly relied on addressing ${\cal P}_{1}^{\bq}$. Moreover, the problem ${\cal P}_{1}^{\bq}$ includes logarithmic form, thus its complexity is $\mathcal{O}\big(L_1(3N)^{3.5}\big)$, where $3N$ is the number of scalar variables and $L_1$ is the number of iterations to update UB trajectory \cite{zhang2019securing}. Then, the overall complexity of Algorithm \ref{Alg1} is $\mathcal{O}\big(L_2L_1(3N)^{3.5}\big)$ where $L_2$
is the number of iterations until convergence.
\begin{figure}[t]
	\centering
	\includegraphics[width=10cm,height=6cm]{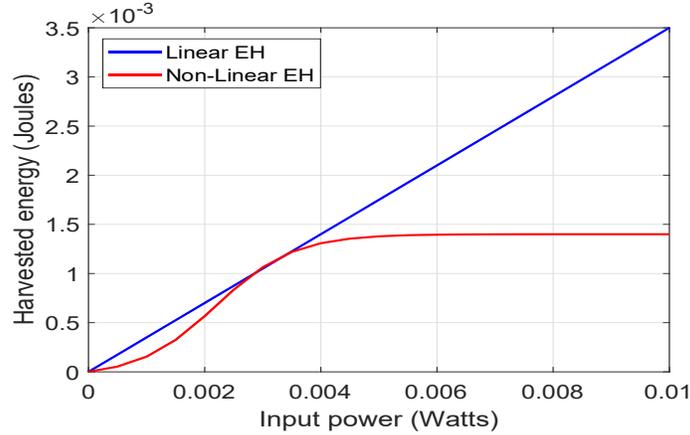}
	\caption { A comparison between the EH models with $\tau=0.5$, $\delta_t=1s$, $\mu=0.7,$ $\Xi=2.8$ mW, $a=1500$, and $b=0.0022$. }
	\label{fig:2}   
\end{figure}

\section{Non-Linear Energy Harvesting (NLEH) Model}
\label{sec:4}
In practice, the RF-to-DC conversion efficiency is usually not a constant value but a non-linear function of the input RF signal \cite{boshkovska2015practical,clerckx2018fundamentals}. As illustrated in Fig. \ref{fig:2}, the harvested energy of both models are well-matched when the received power is low, i.e., input power is less than or equal to 4 mW. When the input power is high, the harvested energy of the linear model is still monotonically increasing while it converges to a saturation value in the non-linear EH model which reveals the limitations of the linear model as compared to that of the non-linear model. Thus, this motivates us to investigate the throughput maximization of cache-aided UAV-enabled backscatter communications under non-linear EH model. Mathematically, the optimization problem is formulated as
\begin{IEEEeqnarray}{rCl}\label{eq:P_2}
	{\cal P}_2: &&\max_{\bq, {\boldsymbol \eta}, {\boldsymbol \tau}, }~~ B \sum\limits_{n \in {\cal N}} \tau_{n} \delta_t \log_2 \Bigg(1+  \frac{ \Theta \eta_u^{n} P_s}{\big({H^2} + {{\left\| {{\bq}_n - {\bw}_{s}} \right\|}^2}\big)^{\alpha/2} \big({H^2} + {{\left\| {{\bq}_n - {\bw}_{d}} \right\|}^2}\big)^{\alpha/2} } \Bigg)  \IEEEyessubnumber \label{eq:P2:a}\\
	\mathtt{s.t.}~~
	&& \eqref{eq:P1:b}, \eqref{eq:P1:c}, \eqref{eq:P1:e}, \eqref{eq:P1:f}, \eqref{eq:P1:g}, \eqref{eq:P1:h}, \IEEEyessubnumber\label{eq:P2:b}\\
	&& \sum\limits_{i=1}^n \tau_n \delta_t P_c \le \sum\limits_{i=1}^n   \frac{(1-\tau_n)\delta_t \Xi}{1-\phi} \Bigg(\frac{1}{1+e^{\frac{-\beta P_s\omega_0}{(H^2+{{\left\| {{\bq}_n - {\bw}_{s}} \right\|}^2})^{\alpha/2}}+\beta\nu} } - \phi \Bigg),
	\IEEEyessubnumber\label{eq:P2:c}
\end{IEEEeqnarray}

The problem ${\cal P}_2$ is still challenging to solve due to the non-convexity of objective function, constraints \eqref{eq:P1:b}, \eqref{eq:P1:c}, and especially the non-linear EH model in constraint \eqref{eq:P2:c} is neither non-convex nor concave. In order to solve ${\cal P}_2$, we propose a iterative algorithm based on BCD method as in Section \ref{sec:3}. Specifically, we decomposed ${\cal P}_2$ into three sub-problems, namely DTS ratio optimization with fixed trajectory and backscatter coefficient, backscatter coefficient optimization with fixed DTS ratio and UAV trajectory, and trajectory optimization with fixed DTS ratio and backscatter coefficient. Consequently, an overall algorithm is proposed to solve each subproblem alternately until convergence is achieved.

\subsection{DTS Ratio optimization:}
For any given value of trajectory $\bq$ and backscatter coefficient $\boldsymbol{\eta}$, the DTS ratio $\boldsymbol{\tau}$ is achieved by solving the following optimization problem:
\begin{IEEEeqnarray}{rCl}\label{eq:P_2t}
	{\cal P}_2^{\boldsymbol{\tau}}: &&\max_{ {\boldsymbol \tau} }~~ B \sum\limits_{n \in {\cal N}} \tau_{n} \delta_t \bar{R}_d^{n}  \IEEEyessubnumber \label{eq:P2t:a}\\
	\mathtt{s.t.}~~
	&& \eqref{eq:P1:b}, \eqref{eq:P1:c}, \eqref{eq:P1:g}, \eqref{eq:P2:c},  \IEEEyessubnumber\label{eq:P2t:c}
\end{IEEEeqnarray}

The optimal value of $\tau_n$ can be obtained as same as Theorem \ref{theorem:1}.
\begin{proposition}
	\label{proposition:1}
	The optimal value $\{\tau_n^\star\}$ to the problem ${\cal P}_2^{\boldsymbol{\tau}}$ can be expressed as
	\begin{subnumcases} {\label{eq:56} \tau_n^\star =}
		\Bigg[\frac{\sigma S}{N \delta_t (\bar{R}^n_d - \bar{R}^n_u)}\Bigg]^1_0, \hfill \; {\rm iff} \bar{R}^n_d > \bar{R}^n_u, \forall n \in {\cal N},
		\label{eq:56a} \\
		\Bigg[\frac{\chi_2 } { \chi_2 +  \delta_t P_c}\Bigg]^1_0. \label{eq:56b}
	\end{subnumcases}
	where $\chi_2 \triangleq  \frac{\delta_t \Xi}{1-\phi} \Bigg(\frac{1}{1+e^{\frac{-\beta P_s\omega_0}{(H^2+{{\left\| {{\bq}_n - {\bw}_{s}} \right\|}^2})^{\alpha/2}}+\beta\nu} } - \phi \Bigg)$.
	\begin{proof}
	The step-by-step to achieve $\tau_n$ can be presented in a similar way to Appendix~A. We only need to replace the linear harvesting model by non-linear harvesting model.
	\end{proof}
	
\end{proposition}

\vspace{-0.2cm}
\subsection{Backscatter Coefficient Optimization:} For any given DTS ratio $\boldsymbol{\tau}$ and trajectory $\bq$, the backscatter coefficient can be achieved by addressing the following optimization problem:
\begin{IEEEeqnarray}{rCl}\label{eq:Pe_2}
	{\cal P}_2^{\boldsymbol{\eta}}: &&\max_{ {\boldsymbol \eta} }~~ B \sum\limits_{n \in {\cal N}} \tau_{n} \delta_t \log_2 \Big(1+  \varphi_1 \eta_u^n \Big)  \IEEEyessubnumber \label{eq:P2e:a}\\
	\mathtt{s.t.}~~
	&& \eqref{eq:P1e:b}, \eqref{eq:P1e:c}, \eqref{eq:P1e:d}. \IEEEyessubnumber\label{eq:P2e:b}
\end{IEEEeqnarray}

It is easy to see that ${\cal P}_2^{\boldsymbol{\eta}}$ is totally the same as ${\cal P}_1^{\boldsymbol{\eta}}$, thus the closed-form expression of $\eta_u^n$ can be obtained as in \eqref{eq:45}.

\vspace{-0.1cm}
\subsection{UB Trajectory Optimization:} For any given values of backscatter coefficient $\boldsymbol{\eta}$ and DTS ratio $\boldsymbol{\eta}$, the UB trajectory can be designed by addressing the following optimization problem:
\begin{IEEEeqnarray}{rCl}\label{eq:P2q}
	{\cal P}_2^{\bq}: &&\max_{\bq }~~ B \sum\limits_{n \in {\cal N}} \tau_{n} \delta_t \log_2 \Bigg(1+  \frac{ \Theta \eta_u^{n} P_s}{\big({H^2} + {{\left\| {{\bq}_n - {\bw}_{s}} \right\|}^2}\big)^{\alpha/2} \big({H^2} + {{\left\| {{\bq}_n - {\bw}_{d}} \right\|}^2}\big)^{\alpha/2} } \Bigg)  \IEEEyessubnumber \label{eq:P2q:a}\\
	\mathtt{s.t.}~~
	&& \eqref{eq:P1:b}, \eqref{eq:P1:c}, \eqref{eq:P1:e}, \eqref{eq:P1:f}, \eqref{eq:P2:c},  \IEEEyessubnumber\label{eq:P2q:b}
\end{IEEEeqnarray}

To make ${\cal P}_2^{\bq}$ more tractable, we introduce slack variables $z_1^n$ and $z_2^n$ as in ${\cal P}_2^{\bq}$. Let us denote ${\bold z} \triangleq \{z_1^n, z_2^n, n\in {\cal N}\}$, then the problem ${\cal P}_2^{\bq}$ is rewritten as
\begin{IEEEeqnarray}{rCl}\label{eq:P2q1.1}
	{\cal P}_{2.1}^{\bq}: &&\max_{\bq, {\bold z} }~~ B \sum\limits_{n \in {\cal N}} \tau_{n} \delta_t \log_2 \Bigg(1+  \frac{ \Theta \eta_u^{n} P_s}{z_1^n z_2^n} \Bigg)  \IEEEyessubnumber \label{eq:P2q1:a}\\
	\mathtt{s.t.}~~
	&& \eqref{eq:P1:e}, \eqref{eq:P1:f}, \eqref{eq:P1q1:b}, \eqref{eq:P1q1:c}, \eqref{eq:P1q1:d}, \eqref{eq:P1q1:e},  \IEEEyessubnumber\label{eq:P2q1:b}\\
	&& \sum\limits_{i=1}^n \tau_n \delta_t P_c \le \sum\limits_{i=1}^n  (1-\tau_n)\delta_t \frac{\Xi}{1-\phi} \Bigg(\frac{1}{1+e^{\frac{-\beta P_s\omega_0}{z_1^n}+\beta\nu} } - \phi \Bigg),
	\IEEEyessubnumber\label{eq:P2q1:c}
\end{IEEEeqnarray}

Similar to Section \ref{subsec:c}, by applying Lemmas \ref{lemma:2} and \ref{lemma:3} to $	{\cal P}_{2.1}^{\bq}$, we have
\begin{IEEEeqnarray}{rCl}\label{eq:P2q2_1}
	{\cal P}_{2.2}^{\bq}: &&\max_{\bq, {\bold z} }~~ B \sum\limits_{n \in {\cal N}} \tau_{n} \delta_t \Theta_2 \IEEEyessubnumber \label{eq:P2q2:a}\\
	\mathtt{s.t.}~~
	&& \eqref{eq:P1:e}, \eqref{eq:P1:f}, \eqref{eq:P1q1:b}, \eqref{eq:P1q1:c}, \eqref{eq:P1q2:c}, \eqref{eq:P1q2:d}, \IEEEyessubnumber\label{eq:P2q2:b} \\
	&& \sum\limits_{i=1}^n \tau_n \delta_t P_c  \le \sum\limits_{i=1}^n  \frac{(1-\tau_n)\delta_t \Xi}{1-\phi} \Bigg(\frac{1}{1+e^{-\beta P_s\omega_0 \big(1/z_1^{n,j}-1/(z_1^{n,j})^2 (z_1^n-z_1^{n,j})\big) +\beta\nu} } - \phi \Bigg).\notag\\
	\IEEEyessubnumber\label{eq:P2q2:c}
\end{IEEEeqnarray}

The problem ${\cal P}_{2.2}^{\bq}$ is much simpler than ${\cal P}_{2.1}^{\bq}$, but it is still troublesome to obtain the direct solution due to the non-convexity of \eqref{eq:P2q2:c}. Let us define $\Theta_3 \triangleq \beta P_s \omega_0 \frac{1}{(z_1^{n,j})^2},$ $\Theta_4 \triangleq -\beta P_s \omega_0 \frac{2}{z_1^{n,j}} +\beta \nu$, the constraint \eqref{eq:P2q2:c} is then represented as
\begin{align}
\label{eq:61}
\sum\limits_{i=1}^n \tau_n \delta_t P_c \le \sum\limits_{i=1}^n (1-\tau_n)\delta_t \frac{\Xi}{1-\phi} \Bigg(\frac{1}{1+e^{\Theta_3 z_1^n + \Theta_4} } - \phi \Bigg).
\end{align}

Let us define $F_1(z_1^n) \triangleq \frac{1}{1+e^{\Theta_3 z_1^n + \Theta_4}}$ and the convexity of $F_1(z_1^n)$ w.r.t $z_1^n$ can be analyzed as follows
\begin{subnumcases} {\label{eq:62} F_1(z_1^n) =}
\rm{concave},\; \rm{if} \; z_1^n < \frac{-\Theta_4}{\Theta_3}, 
\label{eq:62a} \\
\rm{convex},\; \rm{if} \; z_1^n \ge \frac{-\Theta_4}{\Theta_3}.  \label{eq:62b}
\end{subnumcases}

Generally, in the case that $z_1^n < \frac{-\Theta_4}{\Theta_3}$, we then have \eqref{eq:61} is a convex constraint. Thus, ${\cal P}_{2.2}^{\bq}$ can be solved by standard methods \cite{boyd2002}. Besides, $z_1^n \ge \frac{-\Theta_4}{\Theta_3}$, constraint \eqref{eq:61} is non-convex. In order to troublesome this difficulty, we apply SCA method for function $F_1(z_1^n)$ such as
\begin{align}
\label{eq:63}
F_1(z_1^n) \ge \frac{1}{1+e^{\Theta_3 z_1^{n,j} + \Theta_4}} - \frac{\Theta_3e^{\Theta_3 z_1^{n,j} + \Theta_4} \big(z_1^n - z_1^{n,j}\big) }{\big(1+e^{\Theta_3 z_1^{n,j} + \Theta_4}\big)^2 } \triangleq \tilde{F}_1(z_1^n).
\end{align}
	
By substituting \eqref{eq:63} into \eqref{eq:61}, yields
\begin{align}
\label{eq:64}
\sum\limits_{i=1}^n \tau_n \delta_t P_c \le \sum\limits_{i=1}^n (1-\tau_n)\delta_t \frac{\Xi}{1-\phi} \Bigg(\tilde{F}_1(z_1^n) - \phi \Bigg).
\end{align}
 
To this end, by replacing \eqref{eq:64} to \eqref{eq:P2q2:c} in ${\cal P}_{2.2}^{\bq}$ corresponding to the case when $F_1(z_1^n)$ is convex. Thus, we can solve the approximate optimization problem ${\cal P}_{2.2}^{\bq}$ at $j$-th iteration by applying an alternating algorithm which is described as in Algorithm \ref{Alg2}. 

\vspace{-0.3cm}
\begin{algorithm}[h]
	\begin{algorithmic}[1]
		\label{Alg2}
		\protect\caption{Proposed SCA-based Iterative Algorithm to Solve ${\cal P}_2$}
		\global\long\def\algorithmicrequire{\textbf{Initialization:}}
		\REQUIRE  Set $j:=0$ and initialize ${\boldsymbol{\eta}}^j$, ${\boldsymbol{\tau}}^j$, and ${\bq}^j$.
		\vspace{-0.2cm}
		\REPEAT
		\vspace{-0.2cm}
		\STATE Solve ${\cal P}_2^{\boldsymbol{\tau}}$ for given $\{{\boldsymbol{\eta}}^j, {\bq}^j\}$ and denote the optimal solution as $\boldsymbol{\tau}^{j+1}$.
		\vspace{-0.2cm}
		\STATE Solve ${\cal P}_2^{\boldsymbol{\eta}}$ for given $\{{\boldsymbol{\tau}}^{j+1}, {\bq}^j\}$ and denote the optimal solution as $\boldsymbol{\eta}^{j+1}$.
		\vspace{-0.2cm}
		\STATE Solve ${\cal P}_2^{\bq}$ for given $ \{{\boldsymbol{\eta}}^{j+1}, {\boldsymbol{\tau}}^{j+1}\}  $ and denote the optimal solution as $\boldsymbol{\bq}^{j+1}$.
		\vspace{-0.2cm}
		\STATE Set $j:=j+1.$
		\vspace{-0.2cm}
		\UNTIL Convergence. \\
\end{algorithmic} \end{algorithm}


\begin{figure*}[t]
	\centering    
	\subfigure[ $T$ = 6 seconds.] {\label{fig:3a}\includegraphics[width=8cm,height=6cm]{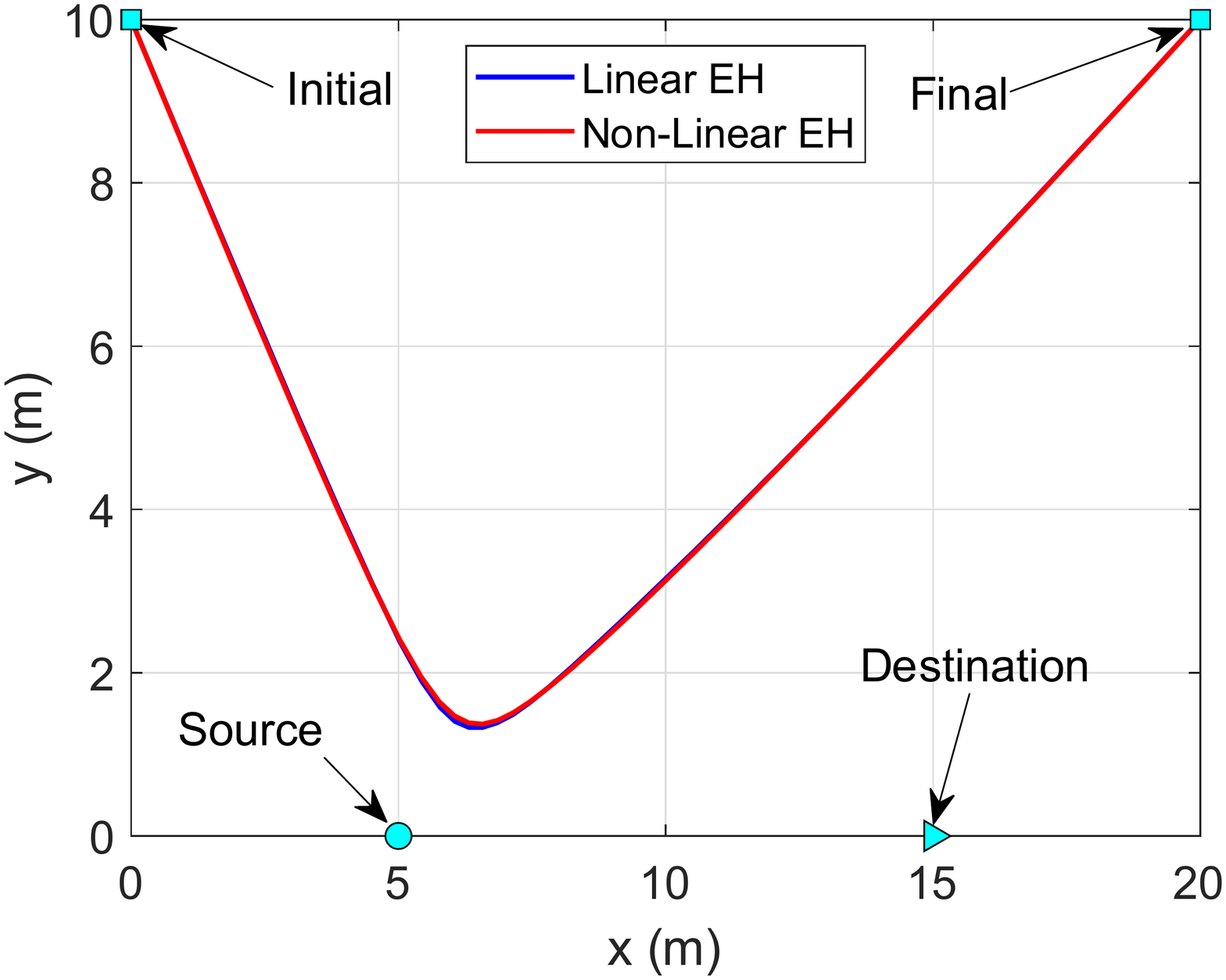}}
	\subfigure[$T$ = 20 seconds.] {\label{fig:3b}\includegraphics[width=8cm,height=6cm]{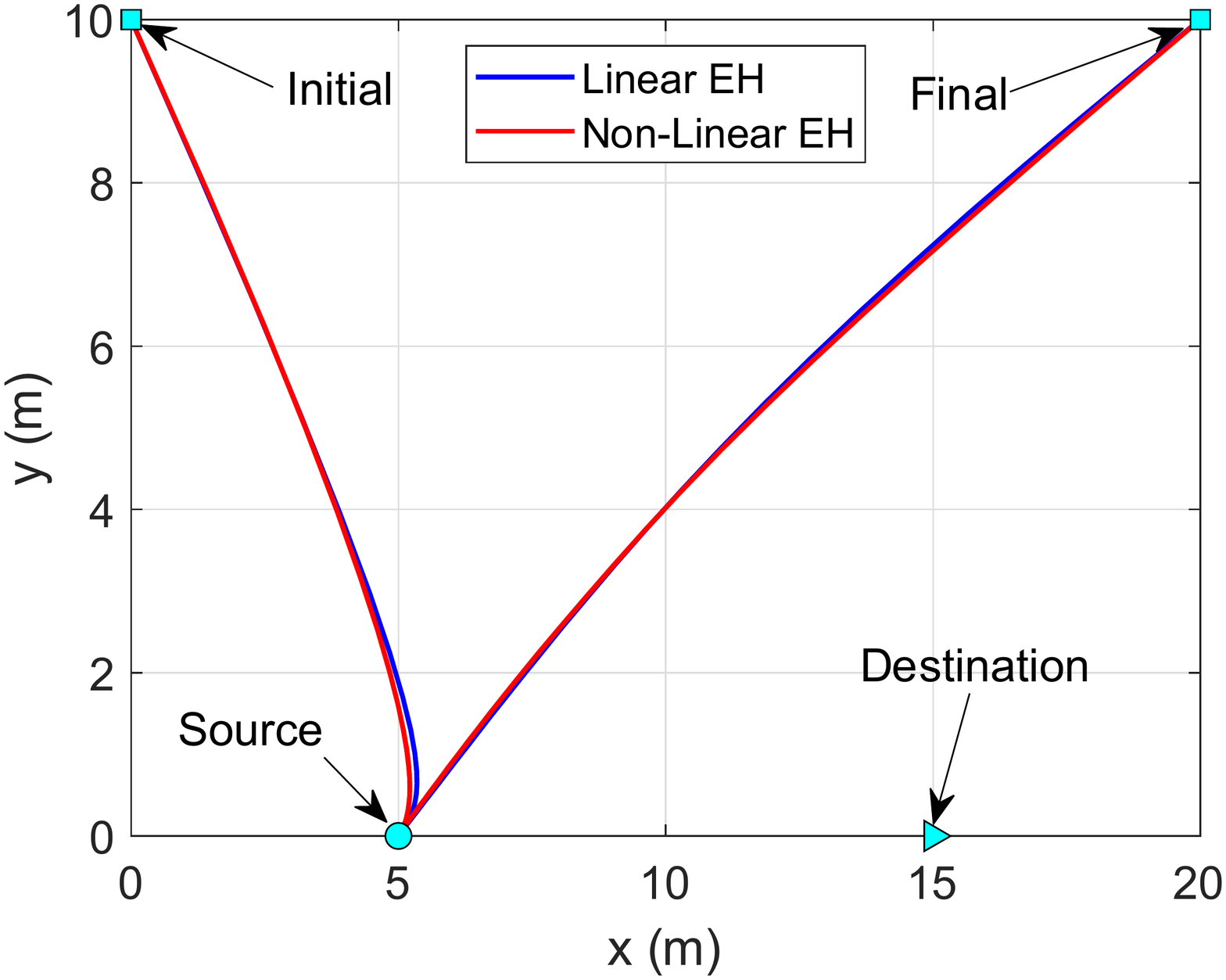}}
	\caption{UB trajectory obtained by our proposed schemes corresponding to linear and non-linear EH models.}
	\label{fig:3}
\end{figure*}

\section{Simulation Results}
\label{Sec:Num}
In this section, the numerical results are given to validate the performance of our proposed schemes under linear and non-linear EH models. We assume that the horizontal locations of the source and destination are set as ${\bw}_s=[5 m,0]^T$ and ${\bw}_s=[15 m,0]^T$, respectively. The UB's initial and final locations are respectively set ${\bq}_{\rm I}=[0,10m]^T$ and ${\bq}_{\rm F}=[20m,10m]^T$. The UB altitude is fixed at $H=10$ meters with maximum transmit power $P=5$ W and maximum velocity $V_{\max}=$ 20 m/s \cite{hua2019}. Moreover, the power channel gain at reference distance $d=1m$ is $-30$ dB \cite{hua2019} and the noise power at the source and destination is $-90$ dB \cite{hua2019}. The circuit power consumption of typical backscatter trasmitter is less than 1 $\mu$W  \cite{lu2018ambient}, thus we set $P_c=10^{-6}$ W. The maximum backscatter coefficient equals to 0.5 \cite{xiao2019resource}. Each time slot duration equals to $0.5$ second and energy harvesting coefficient is 0.9 \cite{hua2019,jayakody2019self}. For non-linear EH model, the maximum harvested power $\Xi=2.8$ mW, EH circuit specifications $\beta=1500$, and $\nu=0.0022$ \cite{kang2019dynamic,Sumit2020Backcom}. The system bandwidth is $B=1$ Mhz. The error tolerance threshold of alternating algorithms is set to $\epsilon=10^{-4}$. To highlight the designed algorithms, we compare our proposed methods with benchmark schemes. Specifically, three benchmark schemes are described as in Table \ref{table:2}. Moreover, LFTra scheme (or NLFTra scheme) is designed similar to that of Algorithm 1 (or 2) but with fixed trajectory, wherein the UB flies from initial position to the middle point between source and destination then it returns to the final position.

\begin{table}[h!]
		\caption{Benchmark Schemes}
		\label{table:2}
		\centering
		{\setlength{\tabcolsep}{0.2em}
			\setlength{\extrarowheight}{0.1em}
			\begin{tabular}{l|l|l}
			\hline \hline
			Scheme & Notation &Descriptions\\
			\hline
			LEH model-based no caching & LNC & Similar to Algorithm \ref{Alg1} but without caching capability \\
			NLEH model-based no caching & NLNC &
			Similar to Algorithm \ref{Alg2} but without caching capability\\
			LEH model-based fixed DTS ratio  &LFTau & Similar to Algorithm \ref{Alg1} but with fixed DTS ratio, $\tau_n=$ 0.5	 \\
			NLEH model-based fixed DTS ratio  & NLFTau & Similar to Algorithm \ref{Alg2} but with fixed DTS ratio, $\tau_n=$ 0.5 \\
			LEH model-based fixed trajectory & LFTra & Similar to Algorithm \ref{Alg1} but with fixed trajectory \\
			NLEH model-based fixed trajectory & LFTra & Similar to Algorithm \ref{Alg2} but with fixed trajectory \\	
				\hline	
				\hline			
	\end{tabular}}
\end{table}


Fig. \ref{fig:3} illustrates the UB trajectories obtained for linear and non-linear EH models at different traveling time, i.e., $T$ equals to 6 seconds and 20 seconds, and $\delta_t=0.2$ s, $P_s=1$ W. It can be seen that the UB tends to fly from initial point to source node and then return back to the final location in both EH models. Specifically, when the total traveling time of UB is increased from 6 seconds to 20 seconds, UB moves in the direction closer to the source location to improve the total throughput.
\begin{figure*}[t]
	\centering    
	\subfigure[Linear EH model.] {\label{fig:4a}\includegraphics[width=8cm,height=6cm]{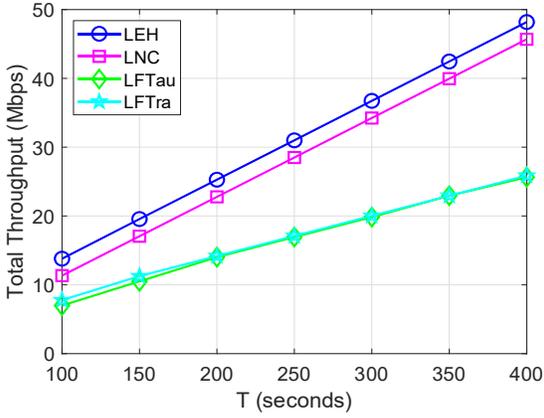}}
	\subfigure[Non-linear EH model.] {\label{fig:4b}\includegraphics[width=8cm,height=6cm]{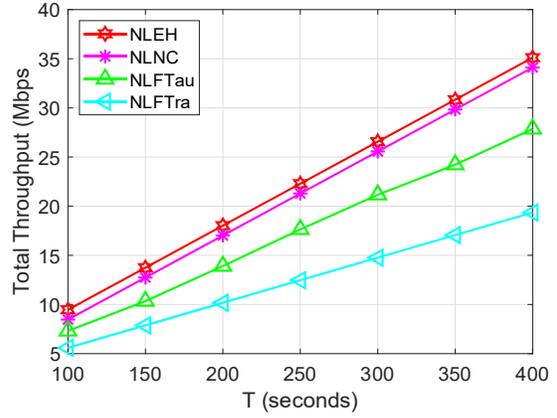}}
	\caption{Total throughput versus traveling time $T$ for the two EH models.}
	\label{fig:4}
\end{figure*}
\begin{figure*}[t]
	\centering    
	\subfigure[Linear EH model.] {\label{fig:5a}\includegraphics[width=8cm,height=6cm]{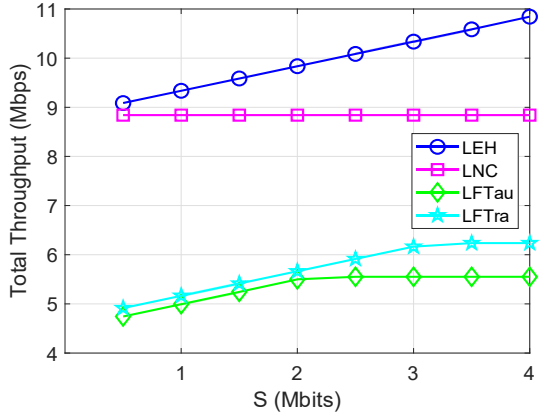}}
	\subfigure[Non-linear EH model.] {\label{fig:5b}\includegraphics[width=8cm,height=6cm]{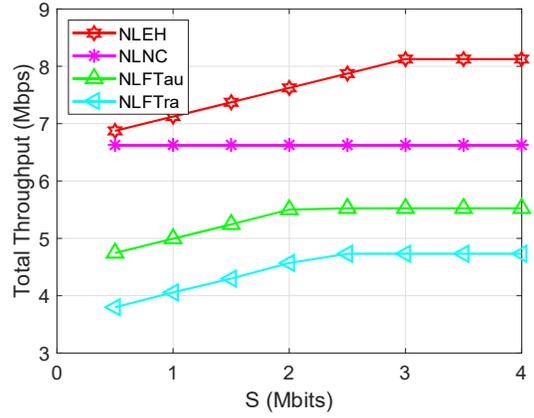}}
	\caption{Total throughput versus demanded data of the destination $S$ for the two EH models.}
	\label{fig:5}
\end{figure*}
\begin{figure*}[t]
	\centering    
	\subfigure[Linear EH model.] {\label{fig:6a}\includegraphics[width=8cm,height=6cm]{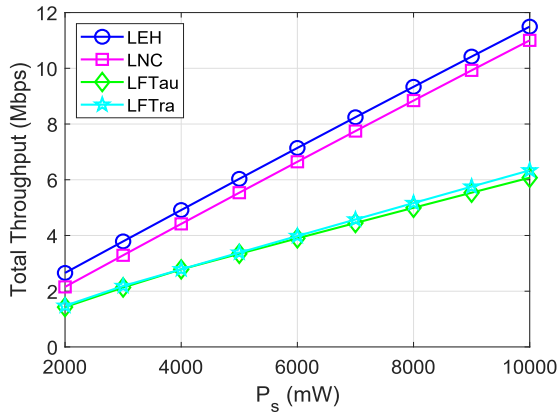}}
	\subfigure[Non-linear EH model.] {\label{fig:6b}\includegraphics[width=8cm,height=6cm]{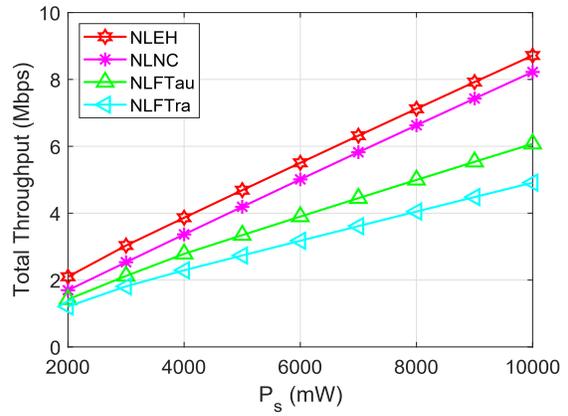}}
	\caption{Total throughput versus transmit power of the source $P_s$ for the two EH models.}
	\label{fig:6}
\end{figure*}
\begin{figure*}[t]
	\centering    
	\subfigure[Linear EH model.] {\label{fig:7a}\includegraphics[width=8cm,height=6cm]{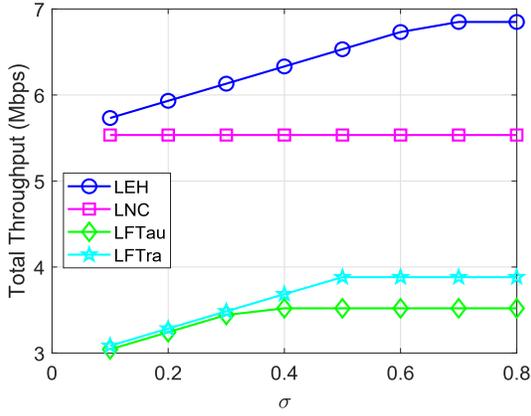}}
	\subfigure[Non-linear EH model.] {\label{fig:7b}\includegraphics[width=8cm,height=6cm]{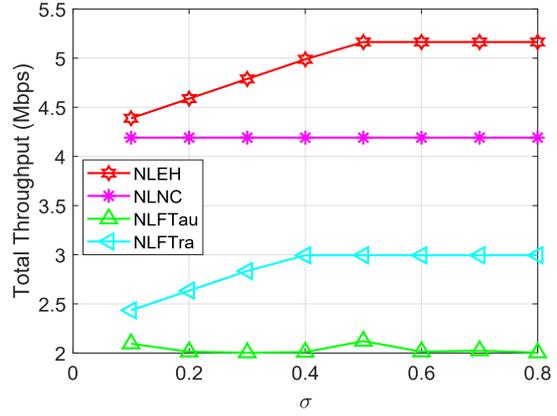}}
	\caption{Total throughput versus $\sigma$ for the two EH models.}
	\label{fig:7}
\end{figure*}
\begin{figure*}[t]
	\centering    
	\subfigure[Linear EH model.] {\label{fig:8a}\includegraphics[width=8cm,height=6cm]{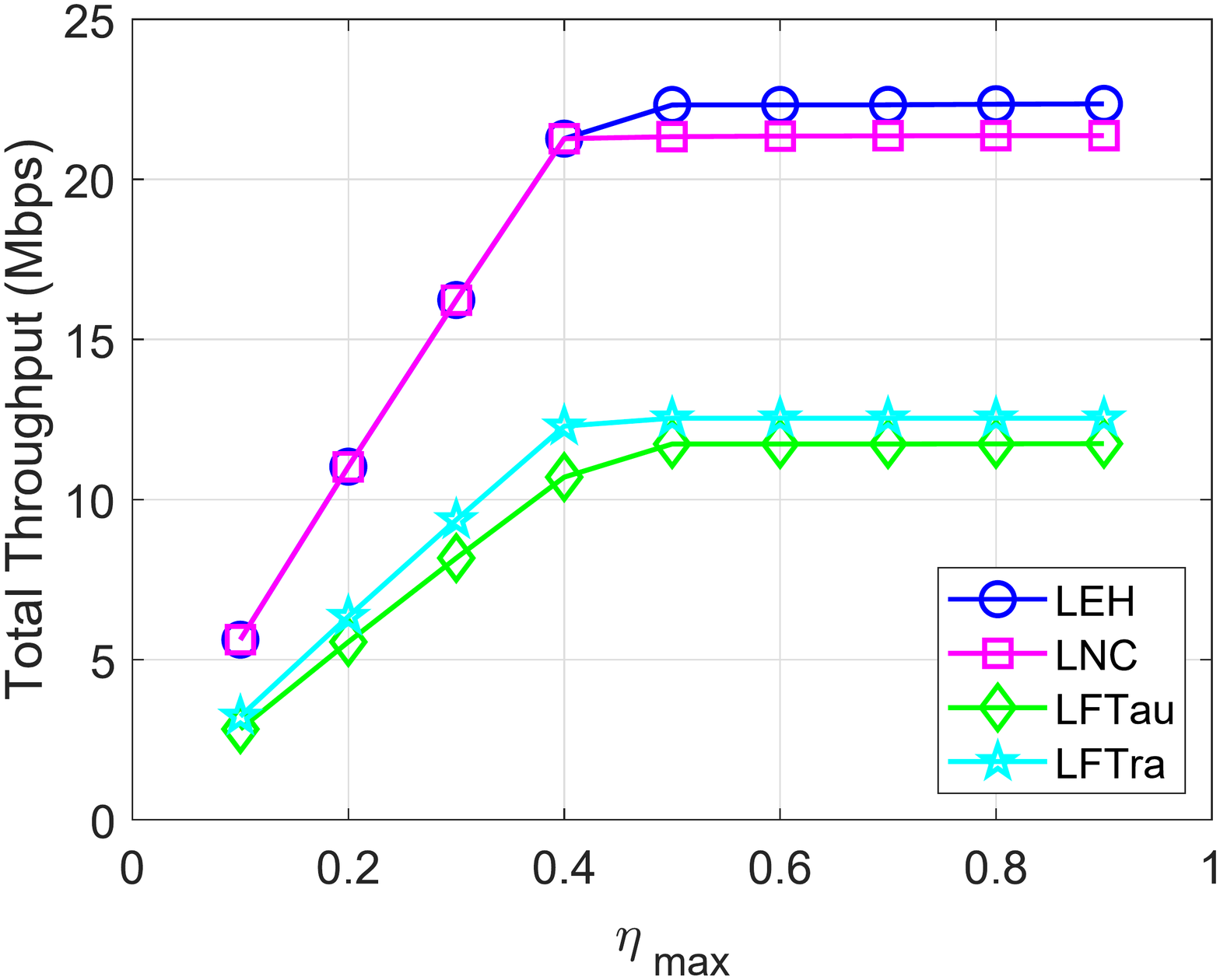}}
	\subfigure[Non-linear EH model.] {\label{fig:8b}\includegraphics[width=8cm,height=6cm]{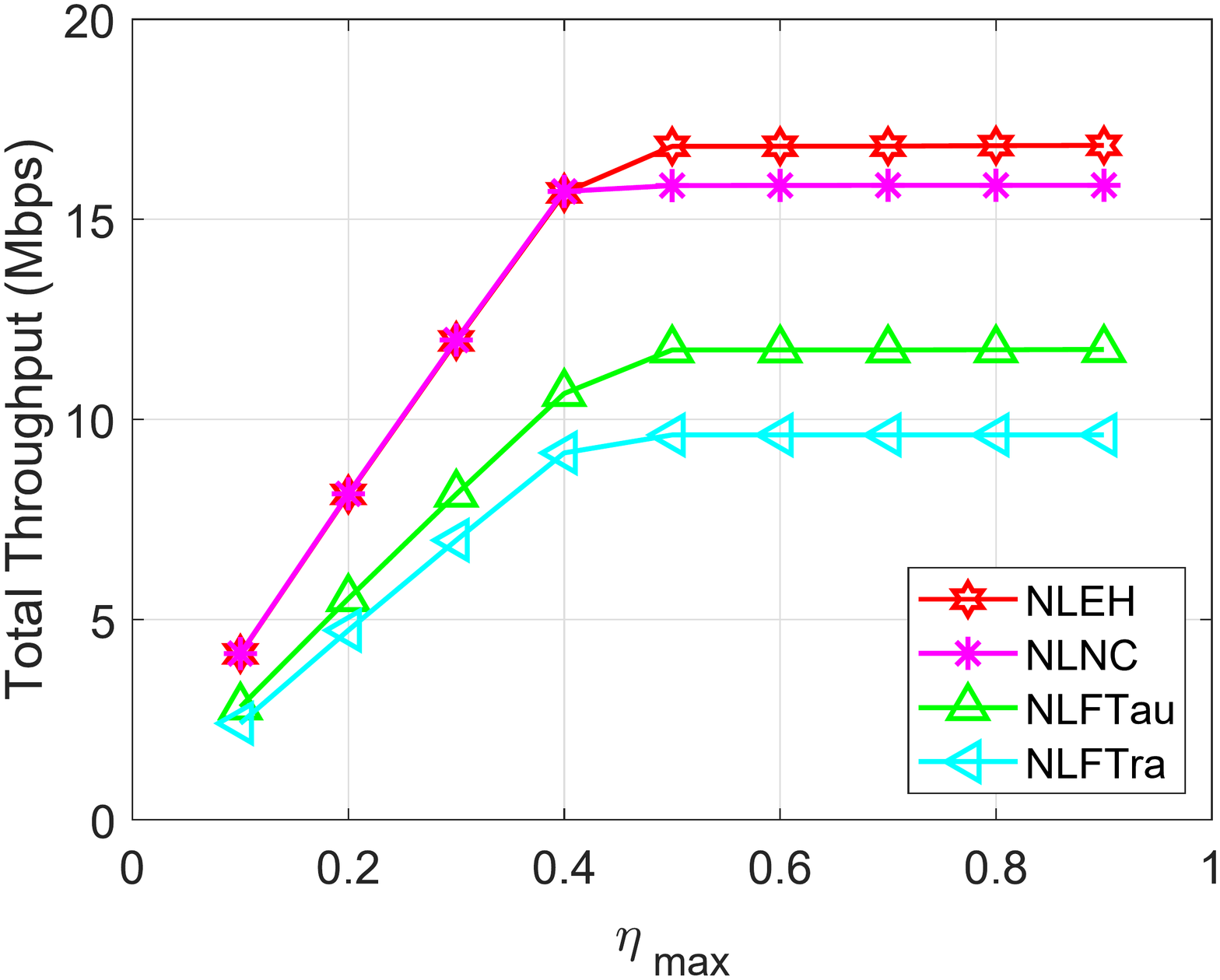}}
	\caption{Total throughput versus $\eta_{\max}$ for the two EH models.}
	\label{fig:8}
\end{figure*}

In Fig. \ref{fig:4}, we investigate the influence of total traveling time to the performance with $\delta_t=$ 0.5 second, $P_s=$ 10 W, $S$ = 2 Mbits. We observe that all the algorithms are linearly increasing with a higher number of traveling time $T$ (in seconds). This is expected since the total collected throughput at the destination is proportional to the reflection time as in \eqref{eq:P1:a}. As inferred from the results, the proposed algorithms significantly improve the total throughput (bps) as compared with the benchmarks. Specifically, in the case of NLEH model at $T=250$ seconds, the NLEH algorithm can support a throughput up to 22.2865 Mbps while the NLNC, NLFTau, NLFTra schemes respectively impose 21.2865, 17.6326, and 12.4695. Moreover, it is also observed that the LEH outperforms the NLEH methods. This is due to the fact that the harvested energy of LEH model is linearly increasing with the input power while it converges to a saturation value in NLEH model. Thus, in the LEH model, the UB can allocate more time for backscattering (i.e., $\tau_n$) as compared with NLEH model which illustrates in \eqref{eq:9a} and \eqref{eq:9b}.

As observed from Fig. \ref{fig:4}, the performance gap between the proposed LEH/NLEH and LNC/NLNC are fixed at a constant value, i.e., $\sigma S$, which does not show the advantages of the proposed schemes with caching capability. This motivates us to plot Fig. \ref{fig:5} which studies the influences of the demanded data $S$ (in bits) to the total throughput, with $T=$ 100 seconds, $P_s=$ 8 W, $V_{\max}=10$ m/s, $H=8$ m. We see that the performance of LNC/NLNC is unaltered by increasing the demanded data $S$. This can be explained that since LNC/NLNC does not cache a part of requested file for data transmission to destination, we then have $\sigma S=0$. Therefore, increasing $S$ value does not impact on LNC/NLNC scheme. This leads to our proposed algorithms, i.e., LEH and NLEH, perform much better than benchmark ones. Particularly, the performance of NLEH is enhanced to a saturation value by growing $S$. This is because the total throughput is restricted by other resources, i.e., transmit power $P_s$, traveling time $T$, reflection time $\tau_n$.

In Fig. \ref{fig:6}, the total throughput is presented as a function of the transmit power of the source $P_s$, where $T=$ 100 seconds, $V_{\max}=10$ m/s, $H=8$ m, $S=1$ Mbits. It is observed that the source transmit power has dramatic impact on the total throughput obtained at the destination. Specifically, the obtained throughput of LEH scheme increases from 2.656 to 11.495 Mbits corresponding to $P_s$ value equals to 2 and 10 W, respectively. Due to the fact that, the total collected throughput depends on the transmit power $P_s$ as shown in Eq. \eqref{eq:P1:a}. Moreover, the more the transmit power is assigned, the higher the harvested energy is achieved. Therefore, the UB can have more reflection time, i.e., higher value of $\tau_n$, which satisfies the energy constraints as in Eqs. \eqref{eq:9a} and \eqref{eq:9b}. Consequently, the system throughput is improved.

Fig. \ref{fig:7} depicts the plot of the total throughput as a function of the caching gain coefficient $\sigma$, with $T=$ 100 seconds, $V_{\max}=10$ m/s, $P_s$ = 5 W, $H=8$ m, $S=2$ Mbits. It is shown from the results that the proposed algorithms significantly enhance the total throughput compared with the references for all values of $\sigma$. Specifically, at $\sigma=0.4$, the LEH and LNC can convey 6.33 and 5.53 Mbps, respectively. Whereas the benchmark LFTau and LFTra respectively impose 3.52 and 3.68 Mbps. In particular, the performance of LNC/NLNC does not depends on caching gain coefficient $\sigma$ as explained in Fig. \ref{fig:5}.

Fig. \ref{fig:8} presents an evaluation of the total throughput versus the maximum value of backscatter coefficient $\eta_{\max}$,  with $T=$ 100 seconds, $V_{\max}=10$ m/s, $P_s$ = 20 W, $H=8$ m, $S=2$ Mbits. It is observed that the LNC/NLNC obtains the same performance as compared to LEH/NLEH scheme when the $\eta_{\max}$ value is small, i.e., $\eta_{\max} \le 0.4$. Otherwise, the proposed LEH/NLEH can obtain much better throughput than other references when $\eta_{\max}\ge 0.5$. Moreover, all methods increase to a saturation value with a higher value of $\eta_{\max}$. It is because the total throughput not only depends on $\eta_{\max}$ but also on $P_s$, $\tau_n$, and $T$.

\section{ Conclusion}
\label{Sec:Con}
We have investigated the cache-assisted wireless powered UAV-enabled backscatter communciations with linear and non-linear EH models. Specifically, we maximized the total throughput via jointly optimizing DTS ratio, backscatter coefficient, and trajectory. The formulated problem was a form of MINLP which is troublesome to solve. Thus, we proposed efficient alternating algorithm based on BCD method and SCA technique to solve it. Particularly, the optimal DTS ratio (or optimal backscatter coefficient) for a given backscatter coefficient (or DTS ratio) and trajectory was derived in closed-form expression which significantly reduced the complexity of proposed solutions. We illustrated via simulation
results that the proposed methods outperformed reference schemes in term of total throughput. Particularly, when the backscatter coefficient and demanded data of the destination was small we should operate in LNC/NLNC for a simple implementation.
%
%
%
%
%
%



\section*{Appendix~A: Proof of Lemma~\ref{lemma:1}}
\label{Appendix:A}
\renewcommand{\theequation}{A.\arabic{equation}}
\setcounter{equation}{0}

\underline{\textit{Proof for \eqref{eq:Lemma1_1}}}: 
Firstly, we consider a function $f(x)=\mathbb{E}_{X}[\log_2(1+x)]$, $x > 0$. Based on Jensen's inequality for concave function $\log_2(1+x)$, we have 
\begin{align}
\label{eq:A1}
f(x) \le \log_2\big(1+ \mathbb{E}_X[x]\big).
\end{align}

Let us denote $x \triangleq P_s |h_{ud}^{n}|^2$, it yields
\begin{align}
\label{eq:A2}
\mathbb{E}_X[x] &= \mathbb{E} \big[P_s \psi_{su} |{\tilde h}_{su}^{n}|^2 \big] = \frac{P_s \omega_0}{\big({H^2} + {{\left\| {{\bq}_n - {\bw}_{s}} \right\|}^2}\big)^{\alpha/2} }.
\end{align}

By substituting \eqref{eq:A2} into \eqref{eq:A1}, we obtain the equation expression \eqref{eq:Lemma1_1}.

\underline{\textit{Proof for \eqref{eq:Lemma1_2}}}: Secondly, we consider a function $f(x,y)=\mathbb{E}_{X,Y}[\log_2(1+xy)]$, $x > 0$, $y > 0$, whereas $x$ and $y$ are two independent random variables. Based on Jensen's inequality for concave function $\log_2(1+xy)$ with respect to (w.r.t.) $y$, we have
\begin{align}
\label{eq:A3}
f(x,y) \le \mathbb{E}_{X}\big[\log_2\big(1+x \mathbb{E}_Y[y]\big)\big] \triangleq \hat{f}(x,y).
\end{align}

Then, by applying Jensen's inequality for convex function $\log_2(1+e^{\ln x})$ with respect to (w.r.t.) $x$, we have 
\begin{align}
\label{eq:A4}
\hat{f}(x,y) = \mathbb{E}_{X}\big[\log_2\big(1+e^{\ln x} \mathbb{E}_Y[y]\big)\big] \ge \log_2\big(1+e^{\mathbb{E}_X[\ln x]} \mathbb{E}_Y[y]\big) \triangleq {\tilde {f}(x,y) }.
\end{align}

From \eqref{eq:A3} and \eqref{eq:A4}, we see that $\tilde {f}(x,y)$ can serve as an approximation function of $f(x,y)$ but it is not a lower bound or a upper bound of $f(x,y)$. 

Let us denote $x \triangleq \frac{\eta_u^{n} P_s |h_{su}^n|^2 }{\sigma_d^2}$, $y \triangleq |h_{ud}^{n}|^2$, and apply \cite[Eq. 4.331.1]{gradshteyn2014}, we then have
\begin{align}
\label{eq:A5}
\mathbb{E}_X[\ln x] &= \int_{0}^{+\infty} \lambda_{su} e^{-\lambda_{su}x} \ln{x} dx  = - (\ln \lambda_{su} + E) = \ln \frac{\eta_u^{n} P_s  \omega_0 (d_{su}^n)^{-\alpha}}{\sigma_d^2} -E, \\
\label{eq:A6}
\mathbb{E}_Y[y] &= \mathbb{E} \big[\psi_{ud} |{\tilde h}_{ud}^{n}|^2 \big] = \frac{\omega_0}{\big({H^2} + {{\left\| {{\bq}_n - {\bw}_{d}} \right\|}^2}\big)^{\alpha/2} },
\end{align}
where $E$ is the Euler–Mascheroni constant, i.e., $E=0.5772156649$ as in \cite[Eq. 8.367.1]{gradshteyn2014}; $\mathbb{E}[x]=\frac{\eta_u^{n} P_s \omega_0 (d_{su}^n)^{-\alpha} }{\sigma_d^2}$, and $\lambda_{su}=\big(\mathbb{E}[x]\big)^{-1}=\big(\frac{\eta_u^{n} P_s \omega_0 (d_{su}^n)^{-\alpha} }{\sigma_d^2}\big)^{-1}$. Substituting \eqref{eq:A5} and \eqref{eq:A6} into \eqref{eq:A4}, we obtain \eqref{eq:Lemma1_2}. Thus, the Lemma \ref{lemma:1} is proof.

\section*{Appendix~B: Proof of Theorem~\ref{theorem:1}}
\label{Appendix:B}
\renewcommand{\theequation}{B.\arabic{equation}}
\setcounter{equation}{0}
It is easy to verify that if $\lambda_4 \ne 0$, thus $J(\boldsymbol{\tau})=0$ implying that $\tau_n=1$ which is not a feasible solution. Thus, we conclude that $\lambda_4=0$. 

In order to obtain the feasible solution, we evaluate all the cases as follows:

\underline{\textit{Case I}}: $\lambda_1=0$ $\implies$ $G({\boldsymbol{\tau}}) \ne 0$, $\lambda_2=0$ $\implies$ $H({\boldsymbol{\tau}}) \ne 0$, $\lambda_3=0$ $\implies$ $I({\boldsymbol{\tau}}) \ne 0$.

From \eqref{eq:B7}, we have $\sum\limits_{n \in {\cal N}}  \delta_t \bar{R}_d^{n}=0$ which is unreasonable. Thus, this case can not occur.

\underline{\textit{Case II}}: $\lambda_1=0$ $\implies$ $G({\boldsymbol{\tau}}) \ne 0$, $\lambda_2 \ne 0$ $\implies$ $H({\boldsymbol{\tau}}) = 0$, $\lambda_3=0$ $\implies$ $I({\boldsymbol{\tau}}) \ne 0$.

From \eqref{eq:B7}, we find that $\lambda_2=-1$ which is unreasonable. Thus, this case can not occur.

\underline{\textit{Case III}}: $\lambda_1\ne 0$ $\implies$ $G({\boldsymbol{\tau}}) = 0$, $\lambda_2 = 0$ $\implies$ $H({\boldsymbol{\tau}}) \ne 0$, $\lambda_3=0$ $\implies$ $I({\boldsymbol{\tau}}) \ne 0$.

From \eqref{eq:B7}, we have $\lambda_1=\frac{\sum\limits_{n \in {\cal N}}\delta_t \bar{R}^n_d }{\sum\limits_{n \in {\cal N}}\delta_t \bar{R}^n_d - \sum\limits_{n \in {\cal N}}\delta_t \bar{R}^n_u }$. If $\bar{R}^n_u = \bar{R}^n_d$, then we obtain $\lambda_1 = +\infty$. If $\bar{R}^n_u > \bar{R}^n_d$, then we obtain $\lambda_1 < 0 $. All of these scenarios is unreasonable. If $\bar{R}^n_u < \bar{R}^n_d$, then we obtain $\lambda_1>1$. Furthermore, from $G({\boldsymbol{\tau}}) = 0$, we have 
\begin{align}
\label{eq_B11}
\tau_n^\star \triangleq \frac{\sigma S}{N \delta_t (\bar{R}^n_d - \bar{R}^n_u)}.
\end{align}

Based on \eqref{eq:B11}, the optimal solution $\{\tau_n^\star\}$ can be obtained iff $\bar{R}^n_u < \bar{R}^n_d, \forall n \in {\cal N}$. 

\underline{\textit{Case IV}}: $\lambda_1 = 0$ $\implies$ $G({\boldsymbol{\tau}}) \ne 0$, $\lambda_2 \ne 0$ $\implies$ $H({\boldsymbol{\tau}}) = 0$, $\lambda_3\ne0$ $\implies$ $I({\boldsymbol{\tau}}) = 0$.

From $H({\boldsymbol{\tau}}) = 0$, we have $\tau_n = \frac{S}{N \delta_t \bar{R}^n_d}$. From $I({\boldsymbol{\tau}}) = 0$, we obtain $\tau_n = \frac{\chi_1 }{ \chi_1 +  \delta_t P_c}$. It can be seen that there exists two different optimal values of $\tau$ which is contradictory. Hence, this case is not occur.

\underline{\textit{Case V}}: $\lambda_1 \ne 0$ $\implies$ $G({\boldsymbol{\tau}}) = 0$, $\lambda_2 = 0$ $\implies$ $H({\boldsymbol{\tau}}) \ne 0$, $\lambda_3\ne 0$ $\implies$ $I({\boldsymbol{\tau}}) = 0$.

\underline{\textit{Case VI}}: $\lambda_1 \ne 0$ $\implies$ $G({\boldsymbol{\tau}}) = 0$, $\lambda_2 \ne 0$ $\implies$ $H({\boldsymbol{\tau}}) = 0$, $\lambda_3 = 0$ $\implies$ $I({\boldsymbol{\tau}}) \ne 0$.

Similar to case IV, we also obtain two different values of $\tau$ in case V and VI which is conflict. Thus, these cases are not occur.

\underline{\textit{Case VII}}: $\lambda_1 \ne 0$ $\implies$ $G({\boldsymbol{\tau}}) = 0$, $\lambda_2 \ne 0$ $\implies$ $H({\boldsymbol{\tau}}) = 0$, $\lambda_3 \ne 0$ $\implies$ $I({\boldsymbol{\tau}}) = 0$.

In this special scenario, we obtain up to three different values of $\tau$ which is unreasonable. Thus, this case is not occur.

\underline{\textit{Case VIII}}: $\lambda_1 = 0$ $\implies$ $G({\boldsymbol{\tau}}) \ne 0$, $\lambda_2 = 0$ $\implies$ $H({\boldsymbol{\tau}}) \ne 0$, $\lambda_3 \ne 0$ $\implies$ $I({\boldsymbol{\tau}}) = 0$.

From \eqref{eq:B7}, we have
\begin{align}
\label{eq:B12}
\sum\limits_{n \in {\cal N}}  \delta_t \bar{R}_d^{n} - \lambda_3 \Big(\chi_1 + \sum\limits_{n \in {\cal N}} \delta_t P_c\Big) = 0 \iff \lambda_3 = \frac{\sum\limits_{n \in {\cal N}}  \delta_t \bar{R}_d^{n} }{\chi_1 + \sum\limits_{n \in {\cal N}} \delta_t P_c}.
\end{align}

Moreover, from $I(\boldsymbol{\tau})=0$, we have
\begin{align}
\label{eq_B2}
 \tau_n^\star = \frac{\chi_1} { \chi_1 +  \delta_t P_c}.
\end{align}

From \eqref{eq_B11}, \eqref{eq_B2}, and constraint \eqref{eq:P1t:e}, we can obtain \eqref{eq:23} which completes the proof of Theorem \ref{theorem:1}.

\vspace{-0.1cm}
\section*{Appendix~C: Proof of Theorem~\ref{theorem:2}}
\label{Appendix:C}
\renewcommand{\theequation}{C.\arabic{equation}}
\setcounter{equation}{0}

We analyze all the possible cases to obtain the feasible solution. The analysis is listing as follows:

\underline{\textit{Case I}}: $\lambda_5=0$ $\implies$ $G_1({\boldsymbol{\tau}}) \ne 0$, $\lambda_6=0$ $\implies$ $H_1({\boldsymbol{\tau}}) \ne 0$.

\vspace{0.2cm}
From \eqref{eq:42}, we have $ \frac{B \sum\limits_{n \in {\cal N}} \tau_{n} \delta_t \varphi_1 }{\ln 2 \Big(1+  \varphi_1 \eta_u^n \Big)}=0$ which implies that $P_s=0$ which is unreasonable. Hence, this case can not occur.

\underline{\textit{Case II}}: $\lambda_5=0$ $\implies$ $G_1({\boldsymbol{\tau}}) \ne 0$, $\lambda_6 \ne 0$ $\implies$ $H_1({\boldsymbol{\tau}}) = 0$.

It can find from \eqref{eq:42} that $\lambda_6 = -1$ which is not a feasible value. Thus, this case can not occur.

\underline{\textit{Case III}}: $\lambda_5\ne 0$ $\implies$ $G_1({\boldsymbol{\tau}}) = 0$, $\lambda_6 \ne 0$ $\implies$ $H_1({\boldsymbol{\tau}}) = 0$.

Combines conditions $G_1({\boldsymbol{\tau}}) = 0$ and $H_1({\boldsymbol{\tau}}) = 0$ with \eqref{eq:40} and \eqref{eq:41}, we can obtain two different optimal values of $\eta_u^n$ which is contradictory. Thus, this case can no occur.

\underline{\textit{Case IV}}: $\lambda_5\ne 0$ $\implies$ $G_1({\boldsymbol{\tau}}) = 0$, $\lambda_6 = 0$ $\implies$ $H_1({\boldsymbol{\tau}}) \ne 0$.


Then, the \eqref{eq:45} can be obtained based on the condition $G_1({\boldsymbol{\tau}}) = 0$ and constraint \eqref{eq:P1e:d} which finishes the proof of Theorem \ref{theorem:2}.

\balance
\bibliographystyle{IEEEtran}
\bibliography{IEEEfull}
\end{document}